%% file: main_arxiv.tex
\documentclass[10pt,twocolumn,twoside]{IEEEtran}

\usepackage{cite}
\usepackage{amsmath,amssymb,amsfonts}
\usepackage{booktabs}
\usepackage{graphicx}
\usepackage{textcomp}
\usepackage{balance}
\usepackage{xcolor}
\usepackage[utf8]{inputenc}
\usepackage{float}
\usepackage{comment}

\usepackage{algorithm}
\usepackage{algorithmic}
\usepackage{framed}
\usepackage[english]{babel}
\usepackage{amsthm}
\usepackage{blindtext}
\usepackage{hyperref}
\usepackage[all]{xy}
\usepackage{tikz}
\usepackage{mdframed}
\usepackage{enumitem}
\usetikzlibrary{arrows}
\usetikzlibrary{positioning,fit,arrows.meta}
\usetikzlibrary{shapes.geometric, shadows, calc}
\usetikzlibrary{matrix, backgrounds, fit, patterns}

\usepackage{titlesec}
\titlespacing{\subsection}
{0pt}    
{0.5ex}  
{1pt}    

\tikzset{
    block/.style = {draw, thick, minimum height=2em, minimum width=3em},
    sum/.style = {draw, circle, node distance=2cm, inner sep=2pt},
    input/.style = {coordinate},
    output/.style = {coordinate},
}

\definecolor{solarPVColor}{RGB}{217, 95, 2}
\definecolor{WindColor}{RGB}{117, 112, 179} %
\definecolor{BESSColor}{RGB}{231, 41, 138}   

\definecolor{orgColor}{RGB}{27, 133, 115}
\definecolor{priceColor}{RGB}{166, 118, 29} %
\definecolor{payColor}{RGB}{50, 160, 50} 
\definecolor{mplgreen}{rgb}{0.173, 0.627, 0.173}
\definecolor{mplred}{rgb}{0.839, 0.153, 0.157}

\definecolor{pinkreward}{RGB}{231, 41, 138}

\newtheoremstyle{compactstyle}
  {3pt}                
  {3pt}                
  {\em}        
  {}                   
  {\bfseries}          
  {.}                  
  {.5em}               
  {}                   

\theoremstyle{compactstyle}
\newtheorem{theorem}{Theorem}

\newtheorem{proposition}[theorem]{Proposition}
\newtheorem{definition}{Definition}[section]

\makeatletter
\renewenvironment{proof}[1][\proofname]{\par
  \pushQED{\qed}%
  \normalfont \topsep0pt \partopsep0pt 
  \trivlist
  \item[\hskip\labelsep
        \itshape
    #1\@addpunct{.}]\ignorespaces
}{%
  \popQED\endtrivlist\@endpefalse
}
\makeatother

\newcommand{\hlight}[1]{%
    \tikz[baseline=(t.base)] {%
        \node (t) [%
            fill=payColor,      
            fill opacity=0.3,   
            text opacity=1,     
            inner sep=2pt,      
            rounded corners=2pt, 
            anchor=base         
        ] {#1};%
    }%
}

\def\BibTeX{{\rm B\kern-.05em{\sc i\kern-.025em b}\kern-.08em
    T\kern-.1667em\lower.7ex\hbox{E}\kern-.125emX}}
\usepackage{lipsum}

\title{\LARGE \bf 
A Coalitional Stable and Fair Reward Allocation for Dynamic Virtual Power Plants}

\author{Carl von Holly-Ponientzietz,  Saverio Bolognani, Florian Dörfler, Verena Häberle%
\thanks{ 
The authors are with the Automatic Control Laboratory, ETH Zurich, 8092 Zurich, Switzerland. Corresponding Email: {\tt verenhae@ethz.ch}}%
\thanks{This work was supported by the Swiss Federal Office of Energy (grant SI/502734 MAESTRO).}%
}

\begin{document}
\begingroup
\allowdisplaybreaks

\maketitle

\begin{abstract}
This paper establishes crucial cooperation criteria for the operation of Dynamic Virtual Power Plants (DVPPs). 
We propose a control design and reward allocation mechanism to enable and incentivize Distributed Energy Resources (DERs) to provide dynamic ancillary services (DAS). Our results illustrate how the cooperative aggregation of heterogeneous DERs leverages technical complementarities to outperform standalone DAS provision.
The proposed reward allocation fulfills critical game-theoretic criteria, including individual rationality, coalitional stability, incentive compatibility, optimality, fairness and ex-post consistency. The control design and reward allocation are validated using a case study based on the Finnish power grid.
\end{abstract}

\begin{IEEEkeywords}
Dynamic Ancillary Services, Cooperative Game Theory, Distributed Energy Resources, Reward Allocation
\end{IEEEkeywords}

\section{Introduction}
The global transition toward renewable energy, driven by climate change mitigation and cost-competitiveness, has resulted in a proliferation of Distributed Energy Resources (DERs). However, the high penetration of these resources creates significant technical and economic challenges. From a technical perspective, the accelerated dynamics of converter-dominated grids requires the deployment of fast ancillary services \cite{arana2020fast}, \cite{pant2024ancillary}. Simultaneously, the economic viability of DERs is increasingly pressured by market saturation. Solar PV and wind assets, for instance, face diminishing returns due to the ``three Cs'' (curtailment, congestion, and price cannibalization) compelling operators to explore alternative revenue streams \cite{curtailment_power_system_planning, wind_cannabilization, nordics_cannabilization}.

Recent work proposes to tackle these challenges by coordinating DERs as  \emph{dynamic virtual power plants (DVPPs)} \cite{DVPP_control}.
DVPPs can be considered as an extension of \emph{virtual power plants (VPPs)}: both VPPs and DVPPs are a collection of DERs, but the DVPP is specifically designed to provide \emph{(fast) dynamic ancillary services (DAS)} while VPPs are often used for demand response \cite{demand_response_vpp}, 
energy arbitrage \cite{vpp_arbitrage}, and power setpoint tracking \cite{vpp_reslience_applications}.
Integrating DERs into a DVPP enables the provision of vital DAS, which can no longer be provided by conventional synchronous generation and simultaneously enhance the economic viability of the participating DERs.

Complementary to the DVPP control design, this paper utilizes cooperative game theory to allocate the reward among the DERs. Game theory has been extensively applied to issues in the power system, as introduced and outlined in \cite{abapour2020game}.
Previous work on VPP reward allocations typically combines a bidding strategy with a (non-)cooperative game-theoretic allocation mechanism. For instance, \cite{profit_allocation_Dabbagh_2015} proposes a risk-based bidding and allocation approach, \cite{Ryu2018GaussianRB} develops a convex game to ensure VPP cooperation in the wholesale market and \cite{bargaining_vpp} applies a \emph{Bargaining Solution} for a BESS-PV system in the ancillary service market. The work in \cite{vpp_allocation_review} reviews VPP optimization and cooperative methods in Section VI-D. 
It identifies the \emph{reward allocation} as the main challenge of cooperative approaches.
Crucially, the existing literature for DER cooperation to provide dynamic ancillary services (DAS) neither establishes suitable criteria for DER cooperation nor addresses the DAS market.

The main contributions of this paper are proposing a joint reward allocation and control design enabling heterogeneous DERs to provide DAS effectively. Our
approach satisfies important cooperative criteria: \emph{individual rationality, coalitional stability, efficiency, optimality, fairness and ex-post consistency}. Intuitively, these criteria imply that the technical capabilities of the DERs are used to their full extent \emph{(optimality)} and the DERs are willing to join and stay in the DVPP \emph{(individual rationality and coalitional stability)}. The formal definitions of the criteria are presented in Sections \ref{sec:problem_setup} and \ref{sec:results}. 

Section \ref{sec:dyn_anc_services} introduces the DAS market framework. The DVPP control, problem formulation, and reward allocation are detailed in Sections \ref{sec:dvpp_framework}--\ref{sec:results}. Finally, Section \ref{sec:case_study} presents a case study based on the Finnish power grid to exemplify the results, followed by the conclusion in Section \ref{sec:conclusion}.

\section{Dynamic Ancillary Services Provision}\label{sec:dyn_anc_services}

Ancillary service markets are generally organized centrally by the system operator. The main market design options, namely the \emph{procurement method}, the \emph{type of price} and the \emph{remuneration structure} \cite{rebours2007survey}, are visualized in Figure \ref{fig:market_summary_compact}.

\begin{figure}[tb]
    \centering
    \resizebox{\linewidth}{!}{%
    \begin{tikzpicture}[
        font=\sffamily,
        node distance=0.15cm,
        title/.style={
            rectangle, rounded corners, draw=black!60, fill=black!5, thick, 
            minimum height=0.6cm, minimum width=8.5cm, 
            align=center, font=\bfseries
        },
        header/.style={
            rectangle, rounded corners, text=white, font=\bfseries\small, 
            minimum height=0.5cm, minimum width=4.1cm, 
            align=center, anchor=north
        },
        box/.style={
            rectangle, draw=black!30, fill=white, rounded corners=2pt, 
            align=left, text width=3.9cm, 
            font=\scriptsize, inner sep=3pt, anchor=north
        },
        fullbox/.style={
            rectangle, draw=black!30, fill=white, rounded corners=2pt, 
            align=left, text width=8.3cm, 
            font=\scriptsize, inner sep=4pt, anchor=north
        }
    ]

    \node[title] (root) {Ancillary Services Market Architecture \cite{rebours2007survey}};

    \node[header, fill=orgColor, below=0.25cm of root.south west, anchor=north west] (org_head) {1. Procurement Method};
    
    \node[box, draw=orgColor, below=0.1cm of org_head] (org1) {
        \textbf{1) Compulsory Provision} \\
        \emph{Con: Ignores costs.}
    };
    \node[box, draw=orgColor, below=0.1cm of org1] (org2) {
        \textbf{2) Bilateral Contracts} \\
        \emph{Con: Low transparency.}
    };
    \node[box, draw=orgColor, below=0.1cm of org2, fill=orgColor, fill opacity=0.3, text opacity=1] (org3) {
        \textbf{3) Auctioning} \\
        Tendering process or spot market.\\
        \emph{Pro: Lowest marginal cost.} \cite{global_ffr_study}
    };

    \node[header, fill=priceColor, below=0.25cm of root.south east, anchor=north east] (price_head) {2. Type of Price};
    
    \node[box, draw=priceColor, below=0.1cm of price_head] (price1) {
        \textbf{1) Regulated Prices}\\
        Fixed by operator. \emph{Con: Inefficient.}
    };
    \node[box, draw=priceColor, below=0.1cm of price1] (price2) {
        \textbf{2) Pay-as-Bid (PAB)}\\
        Paid as submitted bid.
    };
    \node[box, draw=priceColor, below=0.1cm of price2, fill=priceColor, fill opacity=0.3, text opacity=1] (price3) {
        \textbf{3) Clearing / Marginal Pricing (MP)}\\
        \emph{Pro: More transparent and competitive than PAB.} \cite{rancilio2022ancillary}
    };

    \coordinate (lowest_point) at (0,0 |- org3.south); 
    
    \node[header, fill=payColor, minimum width=8.5cm, below=0.15cm of lowest_point, anchor=north] (pay_head) at (root.south |- lowest_point) {3. Remuneration Structure};
    
    \node[fullbox, draw=payColor, below=0.1cm of pay_head] (pay1) {
        \textbf{Possible Payments:} Fixed-allowances $\cdot$ \hlight{Capacity payment} $\cdot$ Utilization payment $\cdot$ Utilization frequency price $\cdot$ Opportunity cost
    };

    \draw[thick, black!50] (root.south) -- +(0,-0.2) -| (org_head.north);
    \draw[thick, black!50] (root.south) -- +(0,-0.2) -| (price_head.north);
    \draw[thick, black!50] (root.south) -- (root.south |- pay_head.north);

    \end{tikzpicture}%
    }
     \vspace{-.5cm}
    \caption{Compact summary of different options of procurement method, type of price and remuneration structure. Our considered DAS market structure is colorcoded, i.e., \textcolor{orgColor}{auctioning}, \textcolor{priceColor}{marginal pricing} and \textcolor{payColor}{capacity payment}.}
    \label{fig:market_summary_compact}
    \vspace{-.4cm}
\end{figure}
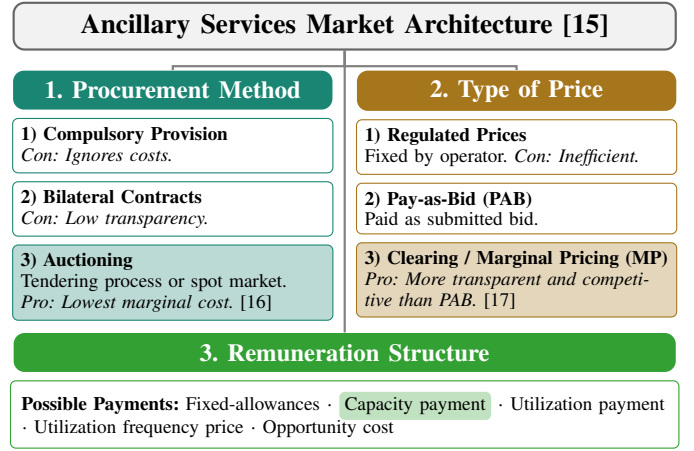

As color-coded in Figure \ref{fig:market_summary_compact}, this paper assumes a competitive DAS \emph{spot market} characterized by \emph{marginal pricing} and \emph{capacity-based remuneration}. This architecture mirrors the Finnish power system used in our case study and was selected for its transparency, economic efficiency, and widespread adoption in DAS markets \cite{ffr_global_original}. Capacity payments and marginal pricing are popular among ENTSO-E members \cite{entsoe2022survey} and are implemented in Australia's \emph{fast raise service} \cite{operator2015guide} and PJM's \emph{dynamic regulation signal} \cite{ffr_global_original}.
Further prominent DAS include Frequency Containment Reserve for Disturbances (FCR-D), voltage control, and Fast Frequency Reserve (FFR). Ireland's EirGrid pioneered FFR in 2018, mandating full capacity within $t_a=2$s of a disturbance \cite{FFR_ireland}. Finland implemented FFR in 2020 ($t_a = 0.7$--$1.3$s) to support low-inertia grids \cite{FFR_finland}. Our proposed method applies to the specific Finnish grid code, but can be easily adapted to conform with other market designs.

We focus on the subset of DAS specified by a fast active power response to frequency disturbances and employ FFR and FCR-D in our case study. However, the methods presented apply to a variety of DAS.

In a majority of today's grid codes, DAS are defined by piecewise linear time-domain curves \cite{european2016commission}. Prominent parameters for this curve include an active power setpoint $\Delta p$, an activation time $t_a$, a deployment time $t_d$ and an end time $t_e$. This is exemplarily visualized in Figure \ref{fig:example_requirement_curve} for Fingrid's Fast Frequency Reserve (FFR) service \cite{FFR_technical}, although oftentimes more parameters are specified.

Individual renewable DERs often struggle to provide DAS due to inherent constraints \cite{global_ffr_study}. As illustrated in Figure \ref{fig:example_requirement_curve}, a specific DER might fail due to energy limits (e.g., a battery energy storage system (BESS), \textcolor{red}{red} response) or ramp-rate constraints (e.g., a wind turbine, \textcolor{orange}{orange} response). However, by aggregating heterogeneous DERs into a Dynamic Virtual Power Plant (DVPP), these complementary limitations can be overcome \cite{DVPP_control}. The DVPP can reliably track the service requirement (\textcolor{mplgreen}{green} response), making cooperation effective. The following section outlines the control design for the DVPP.

\section{Dynamic Virtual Power Plant Control}\label{sec:dvpp_framework}

Let the DVPP $\mathcal{D}$ be a set of controllable DERs cooperatively participating in the DAS market. When the service is activated, the DERs $i\in\mathcal{D}$ respond with an active power deviation $\Delta p_i$ summing up to an \emph{aggregated} deviation; $\Delta p_{\mathrm{agg}}=\sum_{i \in \mathcal{D} }\Delta p_i$.

The DERs are assumed to be located at the same bus to measure a bus frequency deviation $\Delta f$.
The \emph{desired} active power response to frequency can be formulated as:
\begin{equation}
\Delta p_{\mathrm{des}} (s) 
= T_\mathrm{des}(s)
\Delta f(s) 
\end{equation}
where $T_\mathrm{des}(s)$ is the Laplace transformation of the time-domain requirement curve of the service (see \cite{verena_grid_codes} for details), e.g., the blue curve in Figure \ref{fig:example_requirement_curve}.
If, for example, the desired behavior is f-p droop control, then 
$T_\mathrm{des}(s)=\frac{-D}{\tau s+1}$
where $D$ is the desired droop coefficient and $\tau$ is the time constant. 

The \emph{local transfer function} 
$T_i(s)$ of each DER maps the frequency measurement to the DERs active  power deviation:
\begin{equation}
\Delta p_{i} (s) 
= T_i(s)
\Delta f(s). 
\end{equation}

The goal is to guarantee that the aggregated output equals the desired output:
$\Delta p_{\mathrm{agg}} = \Delta p_{\mathrm{des}}$. To ensure that, 
we follow the design procedure proposed in \cite[Section III]{DVPP_control} and previous work \cite{barooah2015spectral}: we divide the DERs into \emph{low-pass filter (LPF)}, \emph{band-pass filter (BPF)} and \emph{high-pass filter (HPF)} roles, depending on each DERs capabilities (e.g. power capacity, time constant, energy limits). The LPF DERs are tasked to cover the long-term steady-state power injection 
over time, while the HPF DERs are responsible for the fast time, pulse-like injections and the BPF DERs cover the intermediate regime. Based on the LPF, BPF and HPF roles, each DER is assigned a portion of the desired transfer function by designing an \emph{adaptive dynamic participation factor (ADPF)} $m_i(s)$, encoding the LPF-, BPF- or HPF-role, respectively:
\begin{equation} \label{equ:local_tf_condition}
    T_i(s) = m_i(s) \cdot T_\mathrm{des}(s), \quad\forall i\in \mathcal{D}.
\end{equation}

The \emph{participation condition} ensures that the overall desired behavior is achieved up to bandwidth $\tau_c$:
\begin{equation} \label{equ:participation_condition}
   \textstyle \sum_{i\in\mathcal{D}}m_i(s) 
    \overset{!}{=} 
    \frac{1}{\tau_cs+1} .
\end{equation}
The ADPFs are time-varying; for instance, a reduction in output at $\hat{i}$ (e.g., due to PV intermittency) triggers a compensatory output increase in DERs $j \in \mathcal{D} \setminus \{\hat{i}\}$.

Given the aggregated DVPP operation is specified as above, the local matching control for each DER $i$ is designed to fulfill the matching condition in equation \eqref{equ:local_tf_condition}. 
Figure~\ref{fig:dvpp_scheme} displays the DVPP control architecture used in the case study of Section \ref{sec:case_study}, where we employed anti-windup-based PI controllers.

Alternatively, more robust and optimal controllers can be employed; for instance, \cite{DVPP_control} uses a $\mathcal{H}_\infty$ local matching control.

\begin{figure}[t!]
    \centering
    \includegraphics[width=.85\linewidth]{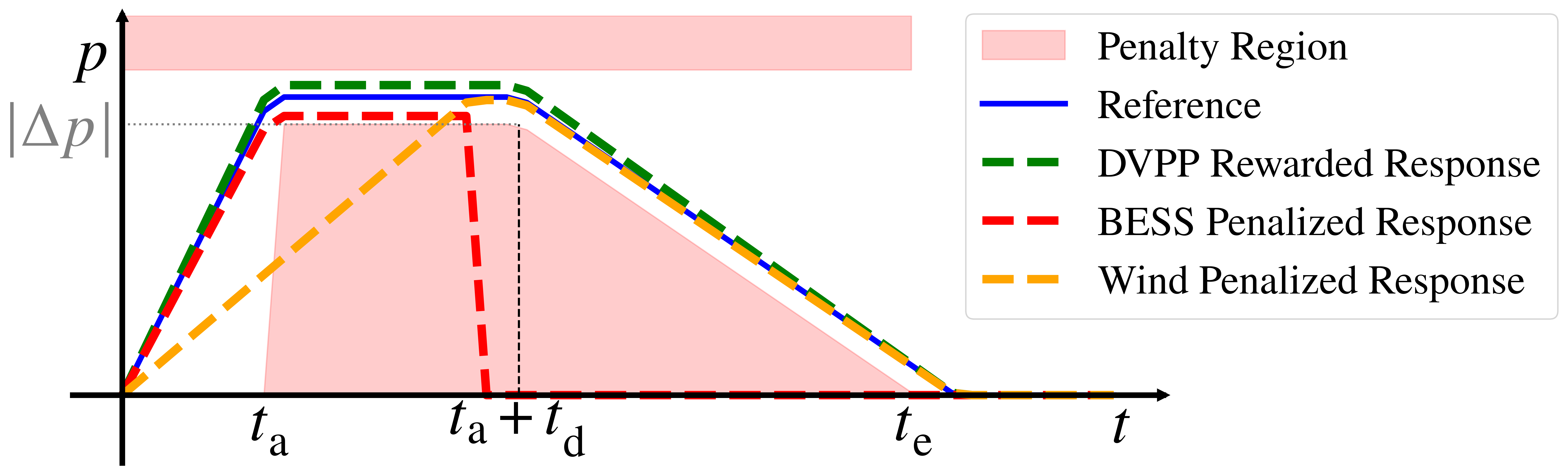}
    \vspace{-.1cm}
    \caption{Example of the time-domain curves for a DAS, namely Fingrid's FFR \cite{FFR_technical}.
    The activation signal is given at $t=0$. Entering the \emph{``Penalty Region''} results in a penalty. The \textcolor{mplgreen}{green} response is rewarded while the \textcolor{red}{red} and \textcolor{orange}{orange} responses are penalized.}\label{fig:example_requirement_curve}
    \vspace{-.2cm}
\end{figure}

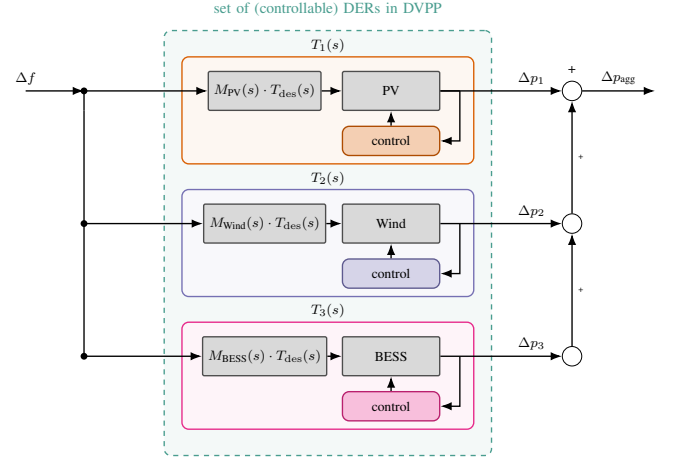
\begin{figure}[t!]
    \centering
    \resizebox{0.48\textwidth}{!}{%
    \input{pics/threeblocks}
    } 
    \caption{DVPP $\mathcal{N}$ control scheme with three heterogeneous controllable DERs, a solar PV power plant (PV), wind power plant (Wind) and battery energy storage system (BESS).}
    \label{fig:dvpp_scheme}
    \vspace{-0.4cm}
\end{figure}

\section{The DVPP Reward Allocation Problem}\label{sec:problem_setup}

This section combines the DVPP control (Section \ref{sec:dvpp_framework}) with the DAS market framework (Section \ref{sec:dyn_anc_services}),  represents the DVPP as a techno-economic agent with coupled control and market objectives, and presents the problem of distributing the economic reward among participating DERs.

Consider a set $\mathcal{D}$ of DERs forming a DVPP. The action following a nominal activation signal $\Delta f$
of a DER $i$ 
is denoted by $u_{i}^{\mathcal{D}}(t) = \Delta p_i(t)$. 
The collective DVPP action is $u_{\mathcal{D}}(t)=\sum_{i\in \mathcal{D}}u_{i}^\mathcal{D}(t)$.

The spot market involves a \emph{bidding} process. A DVPP can submit a \emph{bid capacity} $b(\mathcal{D})$ [MW] for a specific product in the DAS market. We assume that the DVPP
is a price taker and neglect the bid offer price.\footnote{The price-taker assumption is standard for DER aggregators and VPPs
bidding in ancillary-service markets~\cite{iria2019aggregator,
mujeeb2025vpp}. However, we do admit it is a questionable assumption when the DVPP to market size is significant.}

Upon the acceptance of a bid, the DVPP incurs a strict service obligation characterized by a \emph{penalty region} $\rho(t,b)$ (e.g., red shaded area in Figure \ref{fig:example_requirement_curve}), which delimits the permissible power injection envelopes for a submitted bid.
The DVPP employs the control strategy defined in Section~\ref{sec:dvpp_framework} to implement the bid they committed to.

To verify compliance, we construct the \emph{test} function
\begin{equation} \label{equ:test_function}
    h:(\mathcal{D},b) \rightarrow \{\text{fail, pass}\}.
\end{equation}
which compares the DVPP response to $\rho(t,b)$. First, $h(\mathcal{D},b)$ evaluates the response of the DVPP $\mathcal{D}$ to a nominal activation signal $\Delta f$. 
Then, if the action $u_\mathcal{D}(t)$ intersects the penalty region $\rho(t,b)$, the test yields $\text{fail}$; otherwise, it yields $\text{pass}$. 
For example, in Figure \ref{fig:example_requirement_curve}, the standalone BESS and Wind fail the test ($h(\{i\},b)=\text{fail},\,i\in\{\text{BESS, Wind}\}$), 
while the aggregated DVPP passes ($h(\mathcal{\mathcal{D}}_\text{DVPP},b)=\text{pass}$).

We employ the test function $h(\mathcal{D},b)$ to assess the reward of the DVPP:
\begin{align}\label{equ:remuneration}
R(\mathcal{D}, b)
=\begin{cases}
    \pi\cdot b(\mathcal{D}), & \text{if }h(\mathcal{D},b)=\text{pass}\\
    -q\cdot b(\mathcal{D}) & \text{if }h(\mathcal{D},b)=\text{fail}
\end{cases}
\end{align}
where $\pi$ and $q$ are the \textit{clearing} and \textit{penalty} price [€/MW] of the service, respectively\footnote{ We assume that the DVPP operator and the system operator can accurately evaluate $h(\mathcal{D},b)$, i.e., the capacity of the DVPP to supply its bid, using ex-post information. A power grid disturbance is not necessarily required. This is justified in the Finnish case, where Fingrid mandates high-resolution data submission \cite{FFR_finland} and sanctions capacity that is not maintained \cite{FCR_finland}.}.

Within this market framework, one can define the optimal ex-ante selection of the bid magnitude $b(\mathcal{D})$. The DERs $i\in \mathcal{D}$ generally exhibit stochastic parameters, for example Solar PV power production, and thus the attainable service is inherently uncertain. Let $\mathcal{I_D}$ denote the information available to forecast the stochastic parameters of the DVPP. Various tools are available to construct a \emph{probabilistic reward} denoted by 
\begin{equation} \label{equ:generic_prob_reward}
    R_\gamma(\mathcal{D},b,\mathcal{I_D}) 
    = b(\mathcal{D})\cdot[\pi \gamma_{\mathcal{D}}(b) -q (1-\gamma_{\mathcal{D}}(b))],
\end{equation}
where $\gamma_{\mathcal{D}}(b)$ is the probability that the bid $b$ is successfully delivered by DVPP $\mathcal{D}$.
$R_\gamma(\mathcal{D},b,\mathcal{I_D}) $ represents the reward $\mathcal{D}$ can expect ex-ante. For example, if there is a 90\% certainty for $\mathcal{D}$'s capability to provide the bid $\hat{b}$, then $R_\gamma(\mathcal{D},\hat{b},\mathcal{I_D})=\hat{b}\cdot(0.9\cdot\pi  - 0.1\cdot q)$. The \emph{optimal bid} $b^*(\mathcal{D})$ is chosen to maximize (\ref{equ:generic_prob_reward}), i.e.,
\begin{equation}\label{equ:opt_bid}
    b^*(\mathcal{D})\leftarrow \arg\underset{b}{\max} \;
    R_\gamma(\mathcal{D},b,\mathcal{I}).
\end{equation}
The optimal bid $b^*(\mathcal{D})$ is submitted to the market, if it exceeds a required minimum bid capacity $b_\text{min}$. 
At the time of delivery, the system operator evaluates $h(\mathcal{D},b^*)$ to remunerate the DVPP with the reward $R(\mathcal{D}, b^*)$. 

Consider now a superset of DERs $\mathcal N$ such that $\mathcal{D}\subseteq\mathcal{N}$ and denote $\mathcal{N}$ as the \emph{grand coalition}.
As an economic entity, $\mathcal N$ faces the following major decisions:
\begin{enumerate}
    \item How should $\mathcal{N}$ maximize its DAS market value?
    \item How should the received reward be allocated among the DERs $i\in\mathcal{N}$?
\end{enumerate}

To maximize $\mathcal{N}$'s market value, we employ (\ref{equ:generic_prob_reward}) and define the set of all partitions of $\mathcal N$ as $\Pi(\mathcal N)$\footnote{ Formally, let $\mathcal{P}(S)$ be the powerset of $\mathcal{N}$, then \newline$\Pi(S) = \left\{ P \subseteq \mathcal{P}(S) \setminus \{\emptyset\} \;\middle|\; \bigcup_{A \in P} A = S \land \forall A \neq B \in P, A \cap B = \emptyset \right\}$.}.
Let $P^*(\mathcal N)\in\Pi(\mathcal N)$ be the \textit{optimal partition} of DVPPs $\mathcal{D}$ available to $\mathcal{N}$:
\begin{equation}\label{equ:opt_partition}
    P^*(\mathcal N) := \arg\max_{P\in \Pi(\mathcal N)} \textstyle\sum_{\mathcal{D}\in P}R_\gamma(\mathcal{D}, b^*(\mathcal{D}),\mathcal{I}_{\mathcal D}).
\end{equation}

The \emph{forecasted reward} is the aggregated reward the grand coalition can expect by bidding optimally (\ref{equ:opt_bid}) in the market:
\begin{equation}\label{equ:forecast_value_N}
    v_\mathbb{F}(\mathcal N) = 
    \textstyle\sum_{\mathcal{D}\in P^*(\mathcal N)} R_\gamma(\mathcal{D}, b^*(\mathcal{D}),\mathcal{I}_{\mathcal D}).
\end{equation}

Similarly, the \textit{realized reward} is the attained remuneration (\ref{equ:remuneration}) summed over all $\mathcal{D}\in P^*(\mathcal{N})$:
\begin{align}\label{equ:real_reward_value_function}
v_\mathcal{R}(\mathcal N)=\textstyle\sum_{\mathcal{D}\in P^*(\mathcal{N})} R(\mathcal{D},b^*(\mathcal{D}))
\end{align}
However, the DERs $i\in\mathcal{N}$ will only participate in the grand coalition if it is advantageous for them individually. Thus, the \textit{reward allocation} $A^*(\mathcal{N},v)$, yielding \textit{individual payoffs} $x(i)$:
\begin{equation}\label{equ:allocation_rule}
   A^*(\mathcal{N},v):\mathcal{N}\rightarrow  [x(i)]_{i\in \mathcal{N}}, \, \text{ s.t. } \textstyle\sum_{i\in \mathcal{N}}x(i) = v(\mathcal{N})
\end{equation}
is pivotal for the cooperation of the DERs. The allocation must offer each DER $i \in \mathcal{N}$ benefits exceeding their outside options. Critically, the benefits must apply both \emph{before} submitting the bid, $A^*(\mathcal{N},v_\mathbb{F})$, and \emph{after} allocating the reward, $A^*(\mathcal{N},v_\mathcal{R})$, to ensure cooperation. 
Without such guarantees, the grand coalition risks fragmentation as DERs may benefit from exiting the cooperation and forming other coalitions.

\begin{figure}[t!]
    \centering
    \resizebox{0.5\textwidth}{!}{%
    \begin{tikzpicture}[
        >=Latex,
        font=\small,
        eventbox/.style={draw, rectangle, align=center, inner sep=4pt, minimum height=1.1cm, anchor=south},
        timelabel/.style={below=3pt, align=center, font=\footnotesize},
        dashedcyanbox/.style={rectangle, draw=cyan!80, fill=cyan!5, dashed, rounded corners, align=center, font=\small\color{cyan!80}, inner sep=3pt, anchor=center, minimum height=1.1cm},
        dashedpinkbox/.style={rectangle, draw=payColor!80, fill=payColor!5, dashed, rounded corners, align=center, font=\small\color{payColor!80}, inner sep=3pt, anchor=center, minimum height=1.1cm},
        dvpparrow/.style={->, solid, thick, cyan!80, shorten >=2pt, shorten <=2pt},
        rwarrow/.style={->, solid, thick, payColor!80, shorten >=2pt, shorten <=2pt}
    ]

    \coordinate (start) at (0,0);
    \coordinate (pos_dminus1_box) at (1.7, 0);
    \coordinate (pos_submit) at (4.5, 0);
    \coordinate (pos_fcr) at (5.2, 0);
    \coordinate (pos_ffr) at (7.2, 0);
    \coordinate (pos_delivery) at (9.8, 0);
    \coordinate (pos_reward) at (12.5, 0);
    \coordinate (end) at (13.8,0);

    \node[draw, rectangle, anchor=east, minimum size=0.8cm] (d-1) at (start) {D-1};
    \node[draw, rectangle, anchor=west, minimum size=0.8cm] (d+1) at (end) {D+1};
    \draw[->, thick] (d-1.east) -- (d+1.west);

    \node (market_events) [draw=gray, dotted, rounded corners, above=0.25cm of d-1, font=\footnotesize, align=center] {Market Events};
    \draw[dotted, gray, thick] (d-1.north) -- (market_events.south);
    
    \node (dvpp_events_label) [draw=gray, dotted, rounded corners, below=.5cm of d-1, font=\footnotesize, align=center] {Operation\\of DVPP $\mathcal{D}$};
    \draw[dotted, gray, thick] (d-1.south) -- (dvpp_events_label.north);

    \node (reward_mechanism_label) [draw=gray, dotted, rounded corners, below=.7cm of dvpp_events_label, font=\footnotesize, align=center] {Reward\\Mechanism\\for DVPP $\mathcal{N}$};
    \draw[dotted, gray, thick] (dvpp_events_label.south) -- (reward_mechanism_label.north);

    \pgfmathsetlengthmacro{\boxY}{0.2cm}
    \node[eventbox] at (pos_fcr |- 0,\boxY) {FCR\\market\\bids close};
    \node[eventbox] at (pos_ffr |- 0,\boxY) {FFR\\market\\bids close};
    \node[eventbox, minimum width=1.8cm] at (pos_delivery |- 0,\boxY) {Delivery\\for D};
    \node[eventbox] at (pos_reward |- 0,\boxY) {Settlement};

    \node[timelabel] at (pos_dminus1_box) {12AM};
    \node[timelabel] at (pos_fcr) {6:30 PM};
    \node[timelabel] at (pos_ffr) {10:30 PM};
    \node[timelabel] at (pos_delivery) {12AM - 11:59PM};

    \node[dashedcyanbox] (dvpp1) at (pos_dminus1_box |- dvpp_events_label) {calculate\\\& optimize\\$R_\gamma(\mathcal{D},\cdot)$};
    \node[dashedcyanbox] (dvpp2) at (pos_submit |- dvpp_events_label) {submit bid\\$b^*(\mathcal{D})$};
    \node[dashedcyanbox] (dvpp3) at (pos_delivery |- dvpp_events_label) {evaluate\\$h(\mathcal{D},b^*)$};
    \node[dashedcyanbox] (dvpp4) at (pos_reward |- dvpp_events_label) {calculate\\reward\\$R(\mathcal{D},b^*)$};

    \node[dashedpinkbox] (rw1) at (pos_dminus1_box |- reward_mechanism_label) {evaluate $v_\mathbb{F}(S)$ \\ $\forall S\subseteq \mathcal{N}$};
    \node[dashedpinkbox] (rw3) at (pos_delivery |- reward_mechanism_label) {evaluate $h(S,b^*)$ \\ and $v_\mathcal{R}(S),$ \\ $\forall S \subseteq \mathcal{N}$};
    \node[dashedpinkbox] (rw4) at (pos_reward |- reward_mechanism_label) {split reward\\$v_\mathcal{R}(\mathcal{N})$\\among $i\in \mathcal{N}$};

    \draw[dvpparrow] (dvpp1) -- (dvpp2);
    \draw[dvpparrow] (dvpp2) -- (dvpp3);
    \draw[dvpparrow] (dvpp3) -- (dvpp4);
    \draw[rwarrow] (dvpp1) -- (rw1);
    \draw[rwarrow] (dvpp3) -- (rw3);
    \draw[rwarrow] (dvpp4) -- (rw4);
    \draw[rwarrow] (rw1) -- (rw3);
    \draw[rwarrow] (rw3) -- (rw4);

    
    \coordinate (zoom_center) at (7.2, -5.5); 

    \node[draw=payColor, dashed, thick, rounded corners, fill=white, inner sep=5pt] (zoombox) at (zoom_center) {
        \footnotesize
        $
        \xymatrix{
            &
            \text{Is }\mathcal{G}_\mathcal{R}\text{ convex?} \ar[ld]^{\text{\normalsize yes}} \ar[rd]_{\text{\normalsize no}} & \\
            {\begin{array}{cc}
                \text{Shapley Value} \\
                \phi(v_\mathcal{R})
            \end{array}}
            &
            &
            {\begin{array}{cc}
                 \text{Nucleolus} \\
                \mathbf{nc} (v_\mathcal{R})
            \end{array}}
        }
        $
    };

    \node[payColor, anchor=north east] at (zoombox.south east) {Definition of $A^*(\mathcal{N},v_\mathcal{R})$};

    \draw[payColor, dotted, ultra thick] (rw4.south west) -- (zoombox.north east);
    \draw[payColor, dotted, ultra thick] (rw4.south east) -- (zoombox.north east);

    \end{tikzpicture}
    }
    \vspace{-.5cm}
    \caption{DVPP daily operation in the DAS market in Fingrid and the proposed reward allocation. The top row are market events. The middle row, \textcolor{cyan}{cyan colored}, is the operation of a DVPP $\mathcal{D}$ in the market, described in Section \ref{sec:problem_setup} and simulated for all $\mathcal{D}\subseteq\mathcal{N}$. The bottom row, \textcolor{payColor}{green colored}, is the reward allocation timeline outlined in Section \ref{sec:results} for the grand coalition $\mathcal{N}$. The bottom center details the logic of the reward allocation $A^*$.}
    \label{fig:timeline}
    \vspace{-0.4cm}
\end{figure}
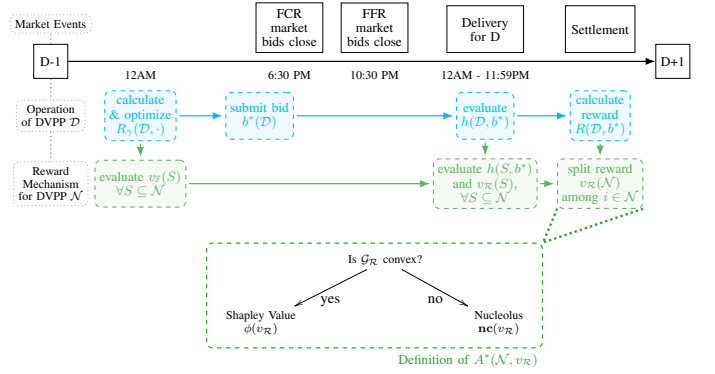

Figure \ref{fig:timeline} illustrates the DAS market timeline.
The \textcolor{cyan}{cyan} row indicates the operation of DVPP $\mathcal{D}$ in the DAS market: submitting the bid $b^*(\mathcal{D})$ based on performance forecasts,
evaluating of the test function $h(\mathcal{D},b^*)$ and receiving the collective reward $R(\mathcal{D},b^*)$. The \textcolor{cyan}{cyan} row is run for each DVPP $\mathcal{D}$ contained in the optimal partition $P^*(\mathcal{N})$. Notice that although the problem is designed for the Finnish market, both Sections \ref{sec:dvpp_framework} and \ref{sec:problem_setup} can be adapted to account for other settings: For instance, the PJM \textit{dynamic regulation signal} requires an open-loop power setpoint (simply replacing $\Delta f\rightarrow r(t)$) and a more complex remuneration (requiring to adjust $h$ and $R$), which does not lead to fundamentally different results.

\section{Results: Proposed reward allocation}\label{sec:results}

In this section, we establish the suitable cooperation criteria and propose a joint control design and reward allocation for the DERs $i\in\mathcal{N}$.
On the premise that these independent DERs seek to participate in DAS provision, the set of fundamental criteria (\ref{C1}-\ref{C5}) governing the DVPP cooperation are:

\begin{mdframed}
\begin{enumerate}[label=\textbf{C\arabic*}]
    \item \label{C1} \text{Individual rationality and coalitional stability}: \emph{Every DER benefits from joining the grand coalition.}
    \item \label{C2} \text{Incentive compatibility}: \emph{It's in the DERs' interest to act truthfully\footnotemark.}
    \item \label{C3} Optimality and Efficiency: \emph{The DVPPs operate optimally and exactly the realized reward is allocated.}
    \item \label{C4} Fairness: \emph{The reward is distributed fairly among its members.}
    \item \label{C5} Ex-Post Consistency: \emph{Ex-post and ex-ante reward align up to a systematic bias $\mu$ over repeated games.}
\end{enumerate}
\end{mdframed}
\footnotetext{ Every DER $i\in\mathcal{N}$ has private information $\mathcal{I}_i$ (e.g., forecast of their own generation, state of charge of battery, ...) which must be shared non-manipulated to enable efficient DVPP cooperation.}
Broadly, Criterion \ref{C1} ensures the \emph{viability}, \ref{C2}--\ref{C3} guarantee the \emph{performance}, and \ref{C4}--\ref{C5} establish the \emph{fairness} of the cooperation.
 
 We propose the following steps (visualized in Figure \ref{fig:timeline}) to fulfill the criteria and answer the problem statement:
 \begin{enumerate}[label=\Alph*., font=\itshape]
    \item \textbf{Scenario Approach:} We employ ensembles of predicted scenarios (e.g., meteorological forecasts) to quantify the probabilistic reward \eqref{equ:generic_prob_reward}.
    \item \textbf{Forecasted Reward Game:} We quantify the value of all DERs and subcoalitions $S\subseteq \mathcal{N}$ by \emph{simulating} their performance in the market as described in Section \ref{sec:problem_setup}. The optimal bid(s) of $\mathcal N$ are submitted to the market.
    \item \textbf{Realized Reward Game:} 
    At delivery time, the capacity of each subcoalition is assessed via the test function. The assessment for $\mathcal N$ determines the realized reward.
    \item[\textit{D.-E.}] \textbf{Reward Allocation:} The reward is allocated to the DERs via $A^*(\mathcal{N},v_\mathcal{R})$, using the Nucleolus or Shapley Value.
\end{enumerate}

\subsection{Scenario Approach}\label{subsec:scenario_approch}
Inspired by previous advances in cooperative game theory, e.g., \cite{ieong2008bayesian}, we select a scenario approach for our issue and employ meteorological forecast data to generate scenarios for the next-day capability of the DERs to deliver the service. 
Let $\mathcal{I}_\mathcal{D}=\{(b_{1},\eta_1),...,(b_{K}, \eta_K)\}$ be a set of $K$ forecasted scenarios for DVPP $\mathcal{D}$, where $\eta_k$ is the probability of each scenario and $b_k$ is the maximum bid that can be delivered in scenario $k$. 

The subset of scenarios capable of achieving a bid $b$ is defined as $\mathcal{M}(\mathcal{D},b)$.
Then, the \emph{reliability} (forecasted probability of success) for bid $b$ is given by
\[
\gamma_{\mathcal{D}}(b)=\textstyle\sum_{k\in \mathcal{M}(\mathcal{D},b)}\eta_k  ,
\]
allowing us to evaluate the probabilistic reward $R_\gamma(\mathcal{D},b,\mathcal{I}_\mathcal{D})$ \eqref{equ:generic_prob_reward} and the optimal bid $b^*(\mathcal{D})$ \eqref{equ:opt_bid}.

\subsection{Forecasted Reward Game}\label{subsec:forecasted_game}

In order to construct a desirable reward allocation, we treat the DERs $i\in\mathcal{N}$ as players in a \emph{cooperative game}. A \textit{value} is assigned to each set of DERs, named a \emph{coalition} $S\subseteq\mathcal{N}$.
\begin{definition}[Cooperative game \cite{osborne1994course}]
\label{def:cooperative_game}
    A cooperative game $(\mathcal{N},v)$ consists of at least the following elements: a set of players $i\in \mathcal{N}$ and a value function $v(S):S\rightarrow \mathbb{R}$ with $v(\emptyset)=0$, which assigns a value to every coalition $S\subseteq\mathcal{N}$.
\end{definition}

For example, $v(\{i\})$ is the value of a standalone DER.
To construct a value function, we \textbf{simulate} the operation of every subcoalition $S\subseteq\mathcal{N}$. Specifically, this involves simulating the optimal partition forming and bidding process outlined in Section \ref{sec:problem_setup} by substituting the grand coalition $\mathcal{N}$ with the hypothetical subcoalition $S$, for all $S\subseteq\mathcal{N}$. The subcoalition operations are purely synthetic operations used for calculation; only the grand coalition $\mathcal{N}$ is physically realized in practice. 

Employing the results from the simulation, we calculate the ex-ante \emph{forecasted value} $v_\mathbb{F}(S)$ (\ref{equ:forecast_value_N}) for all $S\subseteq\mathcal{N}$
and use Definition \ref{def:cooperative_game}, to construct a \emph{Forecasted Reward Game}:
\begin{definition}[Forecasted Reward Game]
    The Forecasted Reward Game denoted by $\mathcal{G}_\mathbb{F}:=\left(\mathcal{N},v_\mathbb{F}\right)$ quantifies the ex-ante value of the coalitions $S\subseteq\mathcal{N}$.
\end{definition}

The forecasted value game $\mathcal{G}_\mathbb{F}$ is evaluated at ``D-1'' in Figure \ref{fig:timeline} and enables the grand coalition formation, submission of bids (\ref{equ:opt_bid}) and operation of DVPPs $\mathcal{D}\in P^*(\mathcal{N})$ (\ref{equ:opt_partition}).

\subsection{Realized Reward Game}\label{subsec:realized_game}

After determining $v_\mathbb{F}(S)$, the ex-post simulation yields the \emph{realized value function} $v_\mathcal{R}(S)$, quantifying the reward (\ref{equ:real_reward_value_function}) by substituting $\mathcal{N}$ with $S,\forall S\subseteq\mathcal{N}$.

We define the \emph{Realized Reward Game} to model the ex-post value of the market settlement. Importantly, the ex-post game involves the absence of agency; the coalition formation is fixed based on the ex-ante results ($\mathcal{G}_\mathbb{F}$) and the DERs simply implement their predefined role in the DVPPs.

\begin{definition}[Realized Reward Game]
    The Realized Reward Game defined by $\mathcal{G}_\mathcal{R}:=\left(\mathcal{N},v_\mathcal{R}\right)$ quantifies the ex-post value for every coalition $S\subseteq \mathcal{N}$.
\end{definition}

\subsection{Nucleolus Reward Allocation}\label{subsec:Nucleolus}
We now propose a reward allocation (\ref{equ:allocation_rule}) that splits the collective reward for both $\mathcal{G}_\mathbb{F}$ and $\mathcal{G}_\mathcal{R}$ via $A^*(\mathcal{N},v_\mathbb{F})$ and $A^*(\mathcal{N},v_\mathcal{R})$, respectively. To satisfy \ref{C1}--\ref{C3}, we first employ the \emph{Nucleolus} as the allocation rule.

The Nucleolus is an allocation method designed to maximize coalitional stability by minimizing the ``dissatisfaction'' over all coalitions of DERs. Dissatisfaction is quantified by the \emph{excess}, $e_S(x) = x(S) - v_\mathbb{F}(S)$, where $x(S)=\sum_{i\in S}x(i)$, measuring $S$'s payoff minus its value. Intuitively, the Nucleolus protects against fragmentation: it identifies the allocation that maximizes the excess of the worst-off coalition, then the second worst-off, and so on. Consequently, it prioritizes that no DER leaves the grand coalition.
Let 
\begin{equation}\label{equ:allocations_ind_rational}
    X_{v_{\mathbb{F}}} = \left\{ x \in \mathbb{R}^{|\mathcal{N}|} \mid x(\mathcal{N})=v_{\mathbb{F}}(\mathcal{N}),\, x(i) \geq v_{\mathbb{F}}(\{i\}) \,\forall i \right\}
\end{equation} 
be the individual rational allocations. Formally, the Nucleolus is the unique allocation $\mathbf{nc}$ that lexicographically maximizes the vector of sorted excesses $\theta(x)$\cite{schmeidler1969nucleolus}:
\begin{equation}\label{equ:nucleolus_def}\mathbf{nc}(v_{\mathbb{F}}) = \left\{ x \in X_{v_{\mathbb{F}}}\mid \theta(x) \succeq_{lex} \theta(y), \forall y \in X_{v_{\mathbb{F}}} \right\}.
\end{equation}
A detailed definition of the Nucleolus is provided in Appendix-\ref{sec:game_theory_definitions}. For our setting, the Nucleolus is the unique reward allocation \cite{nucl_kernel} that ensures grand coalition cooperation at all times:

\begin{proposition}\label{lemma:ex_ante_ind_rational}
    The Forecasted Reward Game $\mathcal{G}_\mathbb{F}$ split with the \emph{Nucleolus} reward allocation fulfills \ref{C1}.
\end{proposition}
\begin{proof}
The game $\mathcal{G}_\mathbb{F}$ is \emph{Superadditive} by definition, meaning the value of a coalition $S_{1}\cup S_2$ is always greater-equal its disjoint subsets, i.e.,
\begin{equation}
\begin{gathered}
    \forall S_1, S_2\subseteq \mathcal{N} \text{ s.t. } S_1\cap S_2=\emptyset: \\
    v_\mathbb{F}(S_1\cup S_2)\geq v_\mathbb{F}(S_1) + v_\mathbb{F}(S_2).
\end{gathered}
\end{equation}
For a Superadditive game, the Nucleolus is individual rational \cite{nucl_kernel}. Moreover, since the game has a \emph{non-empty $\mathrm{core}$} (proof in Appendix-\ref{sec:appendix_proof_non_empty_core}), the Nucleolus lies within the $\mathrm{core}$, thereby providing coalitional stability \cite{nucl_kernel}. 
\end{proof}

\begin{proposition}\label{lemma:incentive_compaibility_nucleolus}
    The Nucleolus method applied to $\mathcal{G}_\mathbb{F}$ and $\mathcal{G}_\mathcal{R}$ ensures incentive compatibility (\ref{C2}). 
\end{proposition}
\begin{proof}
The proof is given in the Appendix-\ref{sec:appendix_proof_c2_nucl}.
\end{proof}

\begin{proposition}\label{lemma:optimality}
    The Nucleolus satisfies \ref{C3}.
\end{proposition}
\begin{proof}
    The bidding strategy is optimized as outlined in
(\ref{equ:opt_bid}) and (\ref{equ:opt_partition}), and the correct parameters are revealed due to incentive compatibility shown in Proposition \ref{lemma:incentive_compaibility_nucleolus}. The Nucleolus also fulfills the efficiency property, i.e., $x(\mathcal{N})=v(\mathcal{N})$.
\end{proof}

\subsection{Shapley Value Reward Allocation} \label{subsec:shapley}

Although the Nucleolus provides a coalitional stable reward allocation under uncertainty \cite{li2015cooperative}, it offers no guarantees to fulfill \ref{C4} and \ref{C5}. To satisfy \ref{C1}-\ref{C5} in particular situations, we introduce the \emph{Shapley Value}, distributing rewards based on a player's average marginal contribution:

\begin{definition}[Shapley Value \cite{shapley1951notes}]
\label{def:shapley}
For a game $(\mathcal{N},v)$, the Shapley Value $\phi(v)=[x(i)]_{i\in\mathcal{N}}$ is defined by
\begin{equation*}
\begin{gathered}
x(i) =\textstyle\sum_{K_i\subseteq\mathcal{N}, i\in K_i}m(|K_i|)\underbrace{\left[v(K_i)-v(K_i \setminus \{i\})\right]}_{\text{marginal contribution}} \,\,, \\
m(|K_i|) = (n-|K_i|)!\cdot(|K_i|-1)!/n!, \quad \forall K_i\subseteq \mathcal{N} 
\end{gathered}
\end{equation*}
\end{definition}

It has unique game-theoretic properties \cite{shapley1951notes} allowing it to attain \ref{C1}-\ref{C5} if the game is \emph{convex}:

\begin{definition}[Convex game]
\label{def:convex_game}
A Game $(\mathcal{N},v)$ is convex iff for all players $i$ and any coalitions $S\subset K\subseteq \mathcal{N}\setminus \{i\}$:
$$v(S\cup \{i\})- v(S)\leq v(K\cup \{i\})- v(K)$$
\end{definition}

Convexity implies that the marginal contribution of a DER $i$ to a coalition increases as the coalition expands.

Conceptually, the Nucleolus sacrifices linearity and fairness in order to preserve the grand coalition $\mathcal{N}$ cooperation (\ref{C1}), compared to the Shapley Value. 

\textbf{To summarize the reward mechanism,} the \textcolor{payColor}{green} row in Figure \ref{fig:timeline} visualizes the ex-ante simulation yielding $v_\mathbb{F}$ to ensure grand coalition cooperation, the execution of the ex-post test function $h(S,b^*),\forall S\subseteq\mathcal{N}$ and finally the reward allocation $A^*(\mathcal{N},v_\mathcal{R})$ programmed to apply the Shapley Value if the game is convex and otherwise the Nucleolus.

\begin{theorem}
    For a convex Forecasted Reward Game, the Shapley Value fulfills criteria \ref{C1}-\ref{C5}\footnote{Although criteria \ref{C1}--\ref{C5} resemble the conditions of the Myerson--Satterthwaite impossibility~\cite{myerson1983efficient}, that result does not apply here: it concerns bilateral trade with \emph{non-verifiable} private valuations and an external budget constraint, whereas in our setting
delivery is verified ex-post via $h(\mathcal{D},b)$ and penalized, and transfers
redistribute the coalition's own realized reward.}.
\end{theorem}
\begin{proof}
The Shapley Value for a convex game lies in the $\mathrm{core}$ \cite{shapley1971cores}, thus ensuring coalitional stability and individual rationality (\ref{C1}). 

The proof for \ref{C2} is presented in Appendix-\ref{sec:appendix_proof_c2_shapley}.

\ref{C3} is met because the Shapley Value fulfills the efficiency property \cite{shapley1951notes} and the DVPP operates optimally via Prop. \ref{lemma:optimality}.

The Shapley Value assigns to each DER its weighted marginal contribution. This allocation is widely recognized as a fairness criterion: it ensures that individual \emph{effort} is proportionally rewarded, given an equal opportunity to participate (a \emph{level playing ground}) \cite{algaba2019shapley, sharing_the_grid, fenton2021fair}. Thus, \ref{C4} is satisfied.

\label{def:bias_mu}
We define ex-post consistency as the  alignment between forecasted and realized rewards over repeated games, accounting for a systematic bias $\mu$ (\ref{C5}).
Assuming the forecasted scenarios have bias $\mu\in\mathbb{R}^{|\mathcal{P}(\mathcal{N})|}$, the expected value of $v_\mathcal{R}$ (\ref{equ:real_reward_value_function}) aligns with the forecasted game, i.e., $\mathbb{E}[v_\mathcal{R}] = v_\mathbb{F} + \mu$. Consequently, criterion \ref{C5} is satisfied due to the linearity of the Shapley Value:
\begin{equation}
    \mathbb{E}[\phi(v_\mathcal{R})]=\phi(\mathbb{E}[v_\mathcal{R}])=\phi(v_\mathbb{F}) + \phi(\mu)
\end{equation}
\end{proof}

\section{Case Study}\label{sec:case_study}

To illustrate the results, we simulate the operation of a set of DERs $\mathcal{N}$ within the Finnish power grid. Finland was selected due to its established FFR market and relevant open-source datasets \cite{fingrid_open_data}.
The units in $\mathcal{N}$ are scaled to approximately 10\% capacity of the existing Ilmatar Hybrid Power Plant \cite{ilmatar2025hybrid}, in particular, a wind power plant (Wind) rated 21.6\,MW, a solar PV plant (PV) rated 15\,MW and a BESS rated 15\,MW.

For a full week, $\mathcal{N}$ participated in 50\% FFR and 50\% $\text{FCR-D}^\text{up}$ control (combined service). The specific timeframe (2025-04-06/12) was selected to include predominantly nonzero FFR market prices. We assume that the DVPPs $\mathcal{D}\in P^*(\mathcal{N})$ must at any point in time be able to react to a large activation signal of $\Delta f=-0.4$ Hz ($f=49.6\text{ Hz}$), requiring a very fast activation time of 1s for FFR and 7.5s for $\text{FCR-D}^\text{up}$ \cite{FFR_technical, fcr_nordics}. As outlined in Section \ref{sec:problem_setup}, the system operator (Fingrid) can simulate whether the constituent DERs are capable of meeting their bid requirement.

The power output of the DERs $i\in\mathcal{N}$ is modeled as a function of DER parameters (power rating, etc.), the environmental conditions (wind speed, solar irradiation, etc.) and the control design (presented in Section \ref{sec:dvpp_framework}). The penalty for non-available capacity is three times the potential reward, i.e. $q=3\pi$ \cite{FFR_finland, fcr_nordics}.
The corresponding parameters are listed in Table \ref{tab:dvpp_params} and DER models in Appendix-\ref{sec:appendix_case_study}. The code for the case study is available on Github. 
\begin{figure}[t!]
    \centering
    \includegraphics[width=1\linewidth]{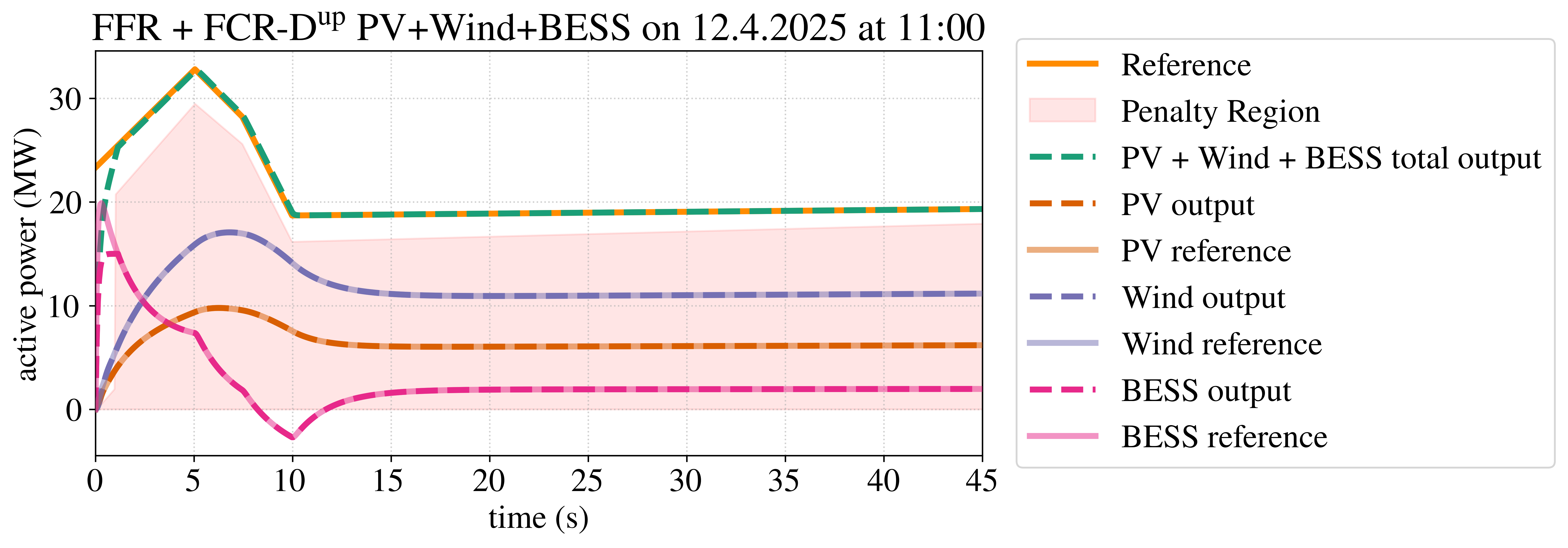}
            \vspace{-.6cm}
    \caption{Performance of the DVPP grand coalition $\mathcal{N}$ on 2025-04-12 at 11:00 to the simulated activation signal $\Delta f=-0.4$ Hz. The \textcolor{mplred}{red-shaded} penalty region represents the requirements for a bid split equally between FFR and $\text{FCR-D}^\text{up}$. The DERs follow their reference nearly perfectly and the test function (\ref{equ:test_function}) is passed, because the total output (\textcolor{mplgreen}{green} line) avoids the penalty region.}
    \label{fig:FFR-FCR_PV+Wind+BESS_DPF_1}
            \vspace{-.4cm}
\end{figure}

\subsection{Illustrative Reward Allocation}
A representative DVPP active power injection profile after a simulated step frequency disturbance is shown in Figure \ref{fig:FFR-FCR_PV+Wind+BESS_DPF_1}. The total output (\textcolor{mplgreen}{green}) must not enter the \textcolor{mplred}{red-shaded} area to avoid a penalty. PV and Wind are selected as low-pass-filter DERs, while the BESS covers the high-pass-filter domain. 

To illustrate the Forecasted and Realized Reward Game, $\mathcal{G}_\mathbb{F}$ and $\mathcal{G}_\mathcal{R}$, for this exact operation, Table \ref{tab:non_empty_core_game} lists the value function on 2025-04-12 at 11:00. The optimization process (\ref{equ:opt_bid}) selects an optimal grand-coalition bid of 33\,MW, submitted equally between FFR and $\text{FCR-D}^\text{up}$. The games are convex (fulfill \ref{def:convex_game}) and the Shapley Value is applied to satisfy criteria \ref{C1}-\ref{C5}.

\begin{table}[t!]
  \centering
  \caption{$v(S)$ for Forecasted ($\mathcal{G}_\mathbb{F}$) and Realized ($\mathcal{G}_\mathcal{R}$) Reward Game in Euros with $\pi=28.5$€/MW on 2025-04-12 at 11:00.}
  \label{tab:non_empty_core_game}
  \resizebox{\columnwidth}{!}{%
    \begin{tabular}{@{}l|lllllll@{}}
    \toprule
     $v(S)$ & PV & Wind & BESS & (PV, Wind) & (PV, BESS) & (Wind, BESS) & $\mathcal{N}$ \\
    \midrule
    $\mathcal{G}_\mathbb{F}$ & 263 & 370 & 107 & 632 & 401 & 618 & 939 \\
    $\mathcal{G}_\mathcal{R}$ & 294 & 444 & 107 & 738 & 401 & 618 & 939 \\
    \bottomrule
    \end{tabular}%
  }
  \vspace{-.2cm}
\end{table}

\begin{table}[t!]
  \centering
  \caption{Shapley Value $\phi$ and (hypothetical) Nucleolus $\mathbf{nc}$ allocation of $\mathcal{G}_\mathbb{F}$ and $\mathcal{G}_\mathcal{R}$ on 2025-04-12 at 11:00.}
  \label{tab:reward_allocation}
  \begin{minipage}[t]{0.48\linewidth}
  \vspace{0pt} 
    \centering
    \begin{tabular}{l|lll}
    \toprule
    $\phi$ & PV & Wind & BESS \\
    \midrule
    $\mathcal{G}_\mathbb{F}$ & 287 & 450 & 202 \\
    $\mathcal{G}_\mathcal{R}$ & 303 & 487 & 149 \\
    \bottomrule
    \end{tabular}
  \end{minipage}\hfill
  \begin{minipage}[t]
  {0.48\linewidth}
    \centering
    \color{black!45}
    \begin{tabular}{l|lll}
    \toprule
    $\mathbf{nc}$ & PV & Wind & BESS \\
    \midrule
    $\mathcal{G}_\mathbb{F}$ & 292 & 454 & 193 \\
    $\mathcal{G}_\mathcal{R}$ & 307 & 484 & 147 \\
    \bottomrule
    \end{tabular}
  \end{minipage}
  \vspace{-0.2cm}
\end{table}

\begin{table}[t!]\label{tab:game_categorization} 
  \centering 
  \caption{Type of allocation used in the period 2025-04-06/12.}
  \label{tab:game_split}
    \begin{tabular}{l|lll}
    \toprule
    & Shapley Value & Nucleolus & Nucleolus \\
    & (convex) & (non-empty $\mathrm{core}$) & (empty $\mathrm{core}$) \\
    \midrule
    $\mathcal{G}_\mathbb{F}$ & 161 & 1 & 0 \\
    $\mathcal{G}_\mathcal{R}$ & 138 & 11 & 13 \\
    \bottomrule
    \end{tabular}%
    \vspace{-0.4cm}
\end{table}

The Shapley Value $\phi$ determines the rewards in the left Table \ref{tab:reward_allocation}. The BESS benefits the most from cooperation compared to acting alone due to its fast response time, allocating a reward $\phi(\text{BESS})$ exceeding its standalone value $v(\{\text{BESS}\})$ by 40\%. Wind captures a higher share of $\mathcal{G}_\mathcal{R}$  compared to $\mathcal{G}_\mathbb{F}$, because it has a higher probability of non-delivery (decreasing $\mathcal{G}_\mathbb{F}$) but could deliver reliably at this specific time (increasing $\mathcal{G}_\mathcal{R}$). A similar $\mathcal{G}_\mathbb{F}/\mathcal{G}_\mathcal{R}$ trend applies to PV but in smaller magnitude because the reliability of PV is higher at this point in time. 

As a comparison, the hypothetical Nucleolus reward allocation is \textcolor{black!45}{listed in greyscale} in Table \ref{tab:reward_allocation}. It allocates slightly more reward to the PV as compared to the Wind and BESS, due to the dissatisfaction minimization objective.

\subsection{Results for 2025-04-06/12}
The hourly average individual value vs. cooperative reward for the entire week is visualized in Figure \ref{fig:values_rewards_DERs}. Every DER gains more from DVPP cooperation compared to standalone operation, resulting in the cooperative reward $v_\mathcal{R}(\mathcal{N})$ surpassing
 the summed individual values of the DERs, $\sum_{i\in\mathcal{N}}v_\mathcal{R}(\{i\})$, by 16\%. The reward is not distributed equally as the device capabilities (power ratings, time constants, etc.) govern the reward allocation.

PV and Wind receive higher rewards in $\mathcal{G_R}$ compared to $\mathcal{G}_\mathbb{F}$, while the opposite is true for the BESS. This can be explained by their probabilistic reward (\ref{equ:generic_prob_reward}) underestimating their real reward (\ref{equ:real_reward_value_function}) due to a forecast error bias $\mu$.

Optimization of $R_\gamma$ (\ref{equ:generic_prob_reward}) for the grand coalition $\mathcal{N}$, \textit{always} yields the most conservative bid selection, i.e., $\gamma_\mathcal{N}(b)=100\%$, due to the large penalty $q=3\pi$. Thus, the bid optimization (\ref{equ:opt_bid}) for $\mathcal{N}$ can be simplified to a risk-averse bidding strategy in our case. This is not true for the subcoalitions $S\subset\mathcal{N}$, where a riskier bidding strategy can be more profitable.

Next, the 162 non-zero price hours of the same week can be categorized into Shapley Value applied and Nucleolus applied, listed in Table \ref{tab:game_split}\footnote{An empty $\mathrm{core}$ arises ex-post when a DER bids risk-affine: its realized standalone value can then
exceed its marginal contribution to a coalition bidding risk-averse, so some
single DER (or sub-coalition) could realize a higher reward. This only
occurs in $\mathcal{G}_\mathcal{R}$ due to ex-ante uncertainty, not in $\mathcal{G}_\mathbb{F}$.}. In a large majority of settings, $\mathcal{G}_\mathbb{F}$ and $\mathcal{G}_\mathcal{R}$ are convex (Definition \ref{def:convex_game}). This can be interpreted as our DVPP of size $|\mathcal{N}|=3$ exhibiting a positive \emph{network effect} \cite{network_effects}, yielding increasing returns as the coalition grows. In the non-convex games (15\% of total), the Nucleolus minimizes the dissatisfaction among the DERs (\ref{equ:nucleolus_def}).

A known limitation of computing the reward allocations is scalability for larger $|\mathcal{N}|$: evaluating $v(S)$ over all $2^{|\mathcal{N}|}$
coalitions is computationally expensive. This is increasingly tractable rather than prohibitive: recent
algorithms compute the Nucleolus efficiently for moderately large games
\cite{benedek2021nucleolus}, and the Shapley Value admits sampling-based
approximation with bounded error \cite{castro2009polynomial}. Scaling the
mechanism to large sets of DERs via these methods is left for future work.

\begin{figure}[t!]
    \centering
    \includegraphics[width=\linewidth]{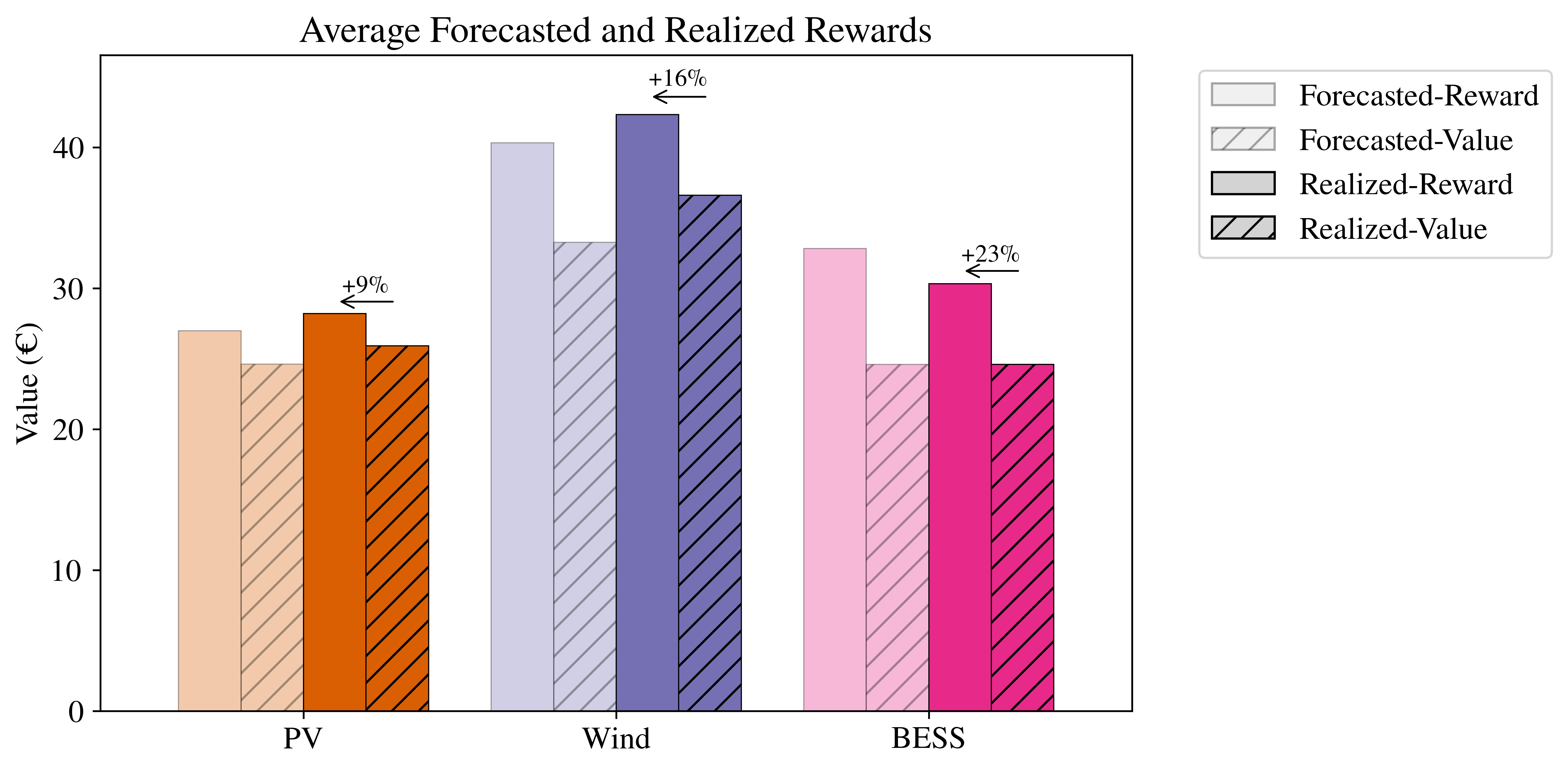}
    \vspace{-2em}
    \caption{Average hourly reward of cooperation (solid bars) vs. value of acting alone (hatched bars) are visualized for the week 2025-04-06/12. The reward of cooperation exceeds the value of acting alone for all DERs and for both the Forecasted and Realized Reward Game, $\mathcal{G}_\mathbb{F}$ and $\mathcal{G}_\mathcal{R}$. The average is computed for all 168h of the week. The arrows $\leftarrow$ indicate the reward percentage increase for every device from cooperation.}
    \label{fig:values_rewards_DERs}
    \vspace{-0.2cm}
\end{figure}

\begin{figure}[t!]
    \centering
    \includegraphics[width=.9\linewidth]{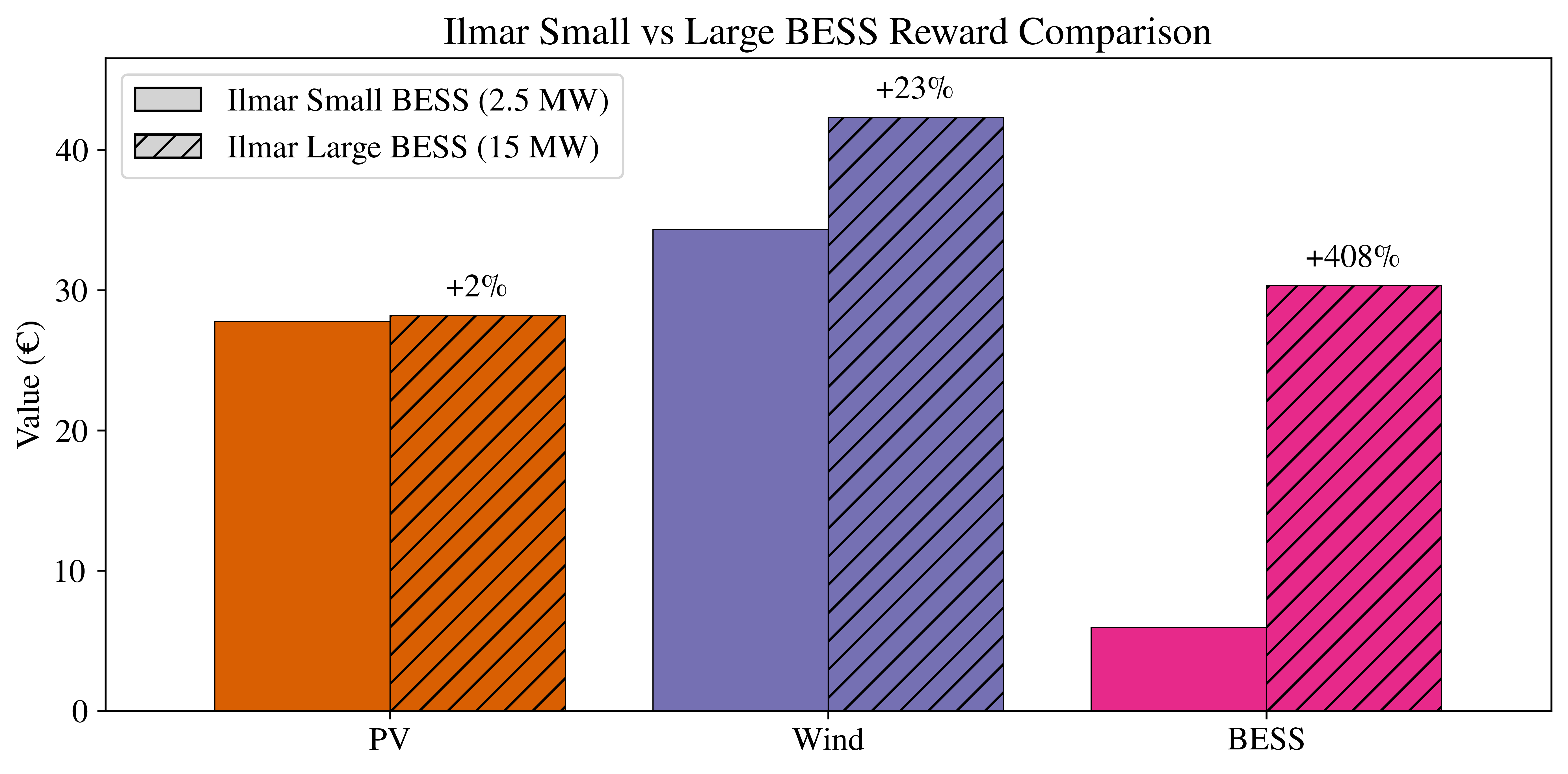}
    \caption{DERs average hourly reward comparison of small (2.5\,MW) vs. large (15\,MW) BESS for the week 2025-04-06/12.}
    \label{fig:Ilmar_vs_large_bess}
            \vspace{-.4cm}
\end{figure}

Increasing DER capacity impacts rewards non-linearly (Figure \ref{fig:Ilmar_vs_large_bess}). Compared to a 6x smaller BESS, the 15\,MW BESS grows its reward by 408\%, but also lifts allocation to Wind by improving Wind's coalitional value through increased reliability. However, PV benefits minimally due to its higher baseline reliability.

\section{Conclusion}\label{sec:conclusion}

This paper presents a reward allocation and control scheme for DERs to cooperatively provide DAS. We established a novel benchmark (\ref{C1}-\ref{C5}) and determined that while criteria \ref{C1}-\ref{C3} are always fulfilled, the convexity of $\mathcal{G}_\mathbb{F}$ and $\mathcal{G}_\mathcal{R}$ is required to satisfy the full set of criteria. The work may incentivize DER cooperation to provide DAS, resulting in both improved power system dynamics and DER economics, as shown by the \textit{cooperative gain} of 16\% compared to standalone operation. Future work can focus on optimizing the ADPF design, integrating multiple services (e.g., FFR, voltage control), investigating larger case studies involving more agents (DERs), more complex grid topologies, and validating larger DVPP configurations.

\section*{Acknowledgments}
The authors acknowledge the use of Gemini during the preparation of this work. 
This tool was employed for spell-checking and debugging suggestions for the codebase. 

\appendix 

\subsection{Game Theory Definitions}\label{sec:game_theory_definitions}

Let $x(i)$ be the \emph{payoff} to DER $i$, then $x=[x(i),...,x(m)]$ is the payoff vector to all DERs and $x(S)=\sum_{i\in S}x(i)$ is the summed payoff to a coalition $S$. 
The allocations satisfying \emph{efficiency} $\left(x(N)= v(N)\right)$ and \emph{individual rationality} $\left(x(i)\geq v(i),\forall i\in\mathcal{N}\right)$ are denoted as $X_v$ and are defined in (\ref{equ:allocations_ind_rational}). The $\mathrm{core}$ is the set of all efficient payoffs that grant coalitional rationality:
\begin{equation}\label{equ:core}
    \mathrm{core}(v):=\left\{x\in X_v: \sum_{i\in S}x(i)\geq v(S),\quad\forall S\subseteq \mathcal{N}\right\}
\end{equation}

To lay out the fundamentals of the Nucleolus, the \emph{excess} of a coalition, $e_S(x)=x(S)-v(S)$, can be interpreted as the "satisfaction" of coalition $S$ (higher $e_S\rightarrow$ higher satisfaction). Let $N'=\{S\mid S\subset \mathcal{N}, S\neq\emptyset\}$ and define the vector $e(x)=[e_S(x)]_{S\subseteq N'}$. Then $e^*(x)$ is the permutation of vector $e(x)$ in increasing order; thus the first entries are the coalition(s) with the highest dissatisfaction (lowest $e_S(x)$). To compare vectors regarding their preference for all coalitions, lexmin superiority can be employed:

\begin{definition}[lexmin superior]
    $e(x)$ is lexmin superior to $e(y)$, denoted by $e(x)\succ_{lex} e(y)$, if $e^*(x)$ is lexicographically  superior to $e^*(y)$ i.e. there exists $l'+1\in\{1,2,...,2^n-2\}$ s.t. $e_l^*(x)=e_l^*(y)$ for $l=1,2,...,l'$ and $e^*_{l'+1}(x) > e^*_{l'+1}(y)$. \\
\end{definition}

With this setup, the Nucleolus can be defined as in equation (\ref{equ:nucleolus_def}) and interpreted as minimizing the maximum dissatisfaction over all coalitions. Importantly, the Nucleolus always lies in the $\mathrm{core}$ if the $\mathrm{core}$ is non-empty \cite{nucl_kernel}.

\subsection{Proof of Non-Empty $\mathrm{core}$ the Forecasted Reward Game}\label{sec:appendix_proof_non_empty_core}

In this subsection, we prove coalitional stability by demonstrating a \textit{non-empty} $\mathrm{core}$, defined in (\ref{equ:core}), for the game $\mathcal{G}_\mathbb{F}$.

First, we define a \textit{balanced set} of weights $\lambda$ for every coalition $S\subseteq\mathcal{N}$, i.e., $\lambda: 2^\mathcal{N} \to [0,1]$, which satisfies $\sum_{S\ni i} \lambda_S = 1$ for all $i \in \mathcal{N}$. Intuitively, a balanced set maps the feasible operation of each DER $i$ to possible subsets $S_i\subseteq \mathcal{N}$ where $i\in S_i$.
Next, we introduce the Bondareva–Shapley Theorem:

\begin{theorem}\label{thm:bondareva_shapley_gen} Bondareva–Shapley Theorem \cite{bondareva1963some, shapley1965balanced} \\
    The $\mathrm{core}$ of a cooperative game $(\mathcal{N}, v)$ is non-empty if and only if the game is balanced. That is, for every balanced set $\lambda$, the following condition holds:
    \begin{equation}\label{equ:balanced_inequality}
        \textstyle\sum_{S \subseteq \mathcal{N}} \lambda_S v(S) \leq v(\mathcal{N}).
    \end{equation}
\end{theorem}

\begin{proof}
Consider the bidding process outlined in Sections \ref{sec:problem_setup} and \ref{sec:results}, and any balanced set $\lambda$. The achieved bids outlined in the Scenario Approach are denoted by $b_k(S)$ for coalition $S$ and scenario $k$. Let $b^*(S)$ be the optimal bid for coalition $S$ and $\gamma_S^*$ the corresponding reliability parameter.

We construct a candidate bid for the grand coalition, $\hat{b}(\mathcal{N})$, as the weighted sum of the optimal sub-coalition bids:
\begin{equation}
    \hat{b}(\mathcal{N}) = \textstyle\sum_{S \subseteq \mathcal{N}} \lambda_S b^*(S).
\end{equation}
This aggregate bid is feasible because $\sum_{S\ni i} \lambda_S = 1$, ensuring the underlying device actions $\hat{u}_i^\mathcal{N}=\sum_{S\subseteq\mathcal{N}}\lambda_S u_i^S$ respect all physical constraints. Due to this feasible operation, we assume the grand coalition $\mathcal{N}$ can manage the aggregated bid $\hat{b}(\mathcal{N})$ with a reliability level $\gamma_\mathcal{N}$ that is at least equal to the capacity-weighted average reliability of the sub-coalitions. Let $\bar{\gamma}$ denote this weighted average:
\begin{equation}
    \bar{\gamma} = \frac{\sum_{S \subseteq \mathcal{N}} \lambda_S b^*(S) \gamma_S^*}{\sum_{S \subseteq \mathcal{N}} \lambda_S b^*(S)}.
\end{equation}
We assume that $\gamma_\mathcal{N}(\hat{b})\geq \bar{\gamma}$. This is physically justified by the risk pooling effect of $\mathcal{N}$ vs. $S,S\subseteq\mathcal{N}$.

Now, we evaluate the weighted sum of the sub-coalition values:
\begin{align*}
    \sum_{S\subseteq \mathcal{N}} \lambda_S v_\mathbb{F}(S) &= \sum_{S\subseteq \mathcal{N}} \lambda_S \left( \pi b^*(S) \gamma_S^* - q b^*(S) (1-\gamma_S^*) \right) \\
    &= \pi \underbrace{\sum_{S\subseteq \mathcal{N}} \lambda_S b^*(S) \gamma_S^*}_{\hat{b}(\mathcal{N}) \cdot \bar{\gamma}} - q \underbrace{\sum_{S\subseteq \mathcal{N}} \lambda_S b^*(S) (1-\gamma_S^*)}_{\hat{b}(\mathcal{N}) \cdot (1-\bar{\gamma})} \\
    &= \hat{b}(\mathcal{N}) \left[ \pi \bar{\gamma} - q (1-\bar{\gamma}) \right].
\end{align*}
This expression represents the reward of the bid $\hat{b}(\mathcal{N})$ with the average reliability $\bar{\gamma}$.
Since the reward function $R_\gamma$ is monotonically increasing with reliability $\gamma$ and using the fact that the forecasted value $v_\mathbb{F}(\mathcal{N})$ \eqref{equ:forecast_value_N} exceeds the candidate probabilistic reward $R_\gamma(\mathcal{N},\hat{b}(\mathcal{N}),\mathcal{I_N})$:
\begin{align*}
    \sum_{S\subseteq \mathcal{N}} \lambda_S v_\mathbb{F}(S)
    &\leq R_\gamma(\mathcal{N}, \hat{b}(\mathcal{N}),\mathcal{I_N}) \Big|_{\gamma=\gamma_\mathcal{N}} \\
    &\leq \sum_{\mathcal{D}\in P^*(\mathcal{N})} R_\gamma(\mathcal{D}, b^*(\mathcal{D}),\mathcal{I}_\mathcal{D}) = v_\mathbb{F}(\mathcal{N}).
\end{align*}
Thus, the condition holds, and the $\mathrm{core}$ is non-empty.
\end{proof}

\subsection{Proof of incentive compatibility of the Nucleolus} \label{sec:appendix_proof_c2_nucl}
\begin{proof}
Let $\tilde{\omega}$ (tilde) denote the variables for the case where DER acts untruthfully.
Assume a DER $j$ manipulates its reported achieved bids $\tilde{b}_k$. Consequently, the bids $\tilde{b}(C),\,j\in C\subseteq \mathcal{N}$ are non-optimal for all coalitions $j$ is a member of. 
This will, in expectation, lead to a lower realized reward for all coalitions $\mathbb{E}[\tilde{v}_\mathcal{R}(C)]< \mathbb{E}[v_\mathcal{R}(C)],\,j\in C\subseteq\mathcal{N}$, because the ex-ante submitted bids are not optimal and thus the ex-post realized performance is worse. 
As a direct effect, a lower overall distributable reward $\tilde{v}_\mathcal{R}(\mathcal{N})<v_\mathcal{R}(\mathcal{N})$ is available in expectation. 
Further, the Nucleolus minimizes the maximum dissatisfaction over all coalitions. By lowering the value $\forall C$, the satisfaction of $C$ with respect to the truthful reward $x(C)$, i.e. $x(C)-\tilde{v}_\mathcal{R}(C)$, increases. 
Recall now that the distributable reward  $\tilde{v}_\mathcal{R}(\mathcal{N})=\sum_{i\in \mathcal{N}}\tilde{x}(i)$ has decreased compared to the truthful case.
This leads to at least one payoff $\tilde{x}(i)$ to decrease compared to the truthful case $x(i)$. 
The decision on which DER payoff to decrease is based on the satisfaction $\tilde{x}(S)-\tilde{v}_\mathcal{R}(S),\forall S\subseteq\mathcal{N}$. Since the satisfaction of all coalitions that $j$ is not a member of, i.e. $K\subseteq\mathcal{N},j\notin K$, did not change, the candidate coalitions for whom to decrease satisfaction are the coalitions $C$. Since $j\in C$, the payoff $\tilde{x}(j)$ decreases to move the satisfaction $\tilde{x}(C)-\tilde{v}_\mathcal{R}(C)$ closer to $x(C)-v_\mathcal{R}(C)$, because of the maximum dissatisfaction the Nucleolus aims to minimize. Thus, $j$ is incentivized to act truthfully\footnote{Intuitively: an untruthful report only shrinks the pie, i.e., it lowers the realized value of every coalition $C\ni j$ while leaving coalitions
without $j$ unchanged. Since the Nucleolus distributes based on coalitional satisfaction $x(S)-v_\mathcal{R}(S)$, and only the coalitions containing $j$ lose value, the burden of the reduced pie falls on $j$'s own payoff. Lying thus hurts efficiency without improving $j$'s relative standing, so $j$ receives the
same or less.}.
\end{proof}

\subsection{Proof of incentive compatibility of the Shapley Value} \label{sec:appendix_proof_c2_shapley}

\begin{proof}
For \ref{C2}, assume a DER $j$ manipulates its reported achieved bids $\tilde{b}_k$. Consequently, the bids $\tilde{b}(C)$ are also manipulated for all coalitions where $j \in C \subseteq \mathcal{N}$. This results in a sub-optimal operation point, leading to a lower expected realized reward $\mathbb{E}[\tilde{v}_\mathcal{R}(C)] < \mathbb{E}[v_\mathcal{R}(C)]$ for all coalitions containing $j$. Crucially, the value of coalitions not containing $j$ remains unaffected (i.e., $\tilde{v}_\mathcal{R}(S) = v_\mathcal{R}(S)$ for $j \notin S$). The Shapley Value $x(j)$ is defined as the weighted average of DER $j$'s marginal contributions to all possible coalitions. Since the value of any coalition including $j$ has decreased, while the value of the coalitions without $j$ remains constant, the marginal contribution term $[\tilde{v}_\mathcal{R}(S \cup \{j\}) - \tilde{v}_\mathcal{R}(S)]$ is strictly smaller than in the truthful case. Consequently, the allocated reward $\tilde{x}(j)$ decreases. Thus, it is in $j$'s interest to act truthfully.
\end{proof}

\subsection{Case Study Appendix}\label{sec:appendix_case_study}

Importantly, we base the DVPP $\mathcal{D}$ reward based on the capacity of the DVPP to supply its capacity bid. This involves simulating a step disturbance on their own frequency measurement and assessing $h(\mathcal{D},b)$ on the DAS requirements: 7.5s activation time for 86\% FCR-$\text{D}^\text{up}$ (1s for FFR) and 20 minutes duration for FCR-$\text{D}^\text{up}$ (5s for FFR) \cite{FFR_finland, fcr_nordics}. 

The parameters of the case study are listed in Table \ref{tab:dvpp_params}. The reference curve and penalty region are taken from Fingrid's specifications \cite{FFR_technical, fcr_nordics}. The functions are as follows:
\begin{align}\label{equ:ffr_tf}\footnotesize
Y^\text{FFR}(t)=\begin{cases}
    1/D_p, & \text{if }t\leq 5\text{s}\\
    1/D_p\cdot\frac{10\text{s}-t}{5\text{s}}, & \text{if }5\text{s}<t\leq 10\text{s}\\
    0 & \text{otherwise}
\end{cases}
\end{align}

\begin{align}\label{equ:fcrd_tf}\footnotesize
Y^{\text{FCR-D}^\text{up}}(t)=\begin{cases}
    1/D_p \cdot \frac{t}{7.5\text{s}}, & \text{if } t \leq 7.5\text{s} \\
    1/D_p \left( 0.86 + 0.19 \frac{t - 7.5\text{s}}{22.5\text{s}} \right), & \text{if } 7.5\text{s} < t \leq 30\text{s} \\
    1/D_p, & \text{if } t > 30\text{s}
\end{cases}
\end{align}

\begin{align}\label{equ:penalty_regions}
& \rho^\text{FFR}(t)=\begin{cases}
    0, & \text{if }t\leq 1\text{s}\\
    1/D_p, & \text{if }1\text{s}<t\leq 5\text{s}\\
    1/D_p\cdot\frac{10\text{s}-t}{5\text{s}}, & \text{if }5\text{s}<t\leq 10\text{s}\\
    0 & \text{otherwise}
\end{cases} \\
\label{equ:penalty_region_fcrd}
&\rho^{\text{FCR-D}^\text{up}}(t)=\begin{cases}
    0.86/D_p \cdot \frac{t}{7.5\text{s}}, & \text{if } t \leq 7.5\text{s} \\
    \left( 0.86 + 0.14 \frac{t - 7.5\text{s}}{22.5\text{s}} \right)/D_p, & \text{if } 7.5\text{s} < t \leq 30\text{s} \\
    1/D_p, & \text{if } t > 30\text{s}
\end{cases}
\end{align}

\begin{table}
    \centering
    \caption{DVPP parameters used in the case study}
    \label{tab:dvpp_params}
    \begin{tabular}{@{}lcc@{}}
        \toprule
        \textbf{Parameter} & \textbf{Symbol} & \textbf{Value} \\
        \midrule
        Power rating, wind & $S_{\text{wind,r}}$ & 21.6 MVA \\
        Power rating, PV & $S_{\text{pv,r}}$ & 15 MVA \\
        Power rating, BESS & $S_{\text{bess,r}}$ & 15 MVA \\
        Energy rating, BESS & $E_{\text{bess,r}}$ & 5 MWh \\
        \midrule
        Online power wind & $\overline{p}_{\text{wind}}(\mathcal{H}_l)$ & $S_{\text{wind,r}}\cdot \kappa_{\text{wind}}(\mathcal{H}_l)$ \\
         Online power PV & $\overline{p}_{\text{pv}}(\mathcal{H}_l)$ & $S_{\text{pv,r}}\cdot \kappa_{\text{pv}}(\mathcal{H}_l)$ \\
        \midrule
        Activation signal & $\Delta f$ & $-0.4$ Hz \\
        FFR transfer function & $T_\mathrm{des}^\mathrm{FFR}(s)$ & $\mathrm{Laplace}\{Y^\mathrm{FFR}(t)\}$\\
        $\text{FCR-D}^\text{up}$ transfer f. & $T_\mathrm{des}^\mathrm{FCR-D^{up}}(s)$ & $\mathrm{Laplace}\{Y^\mathrm{FCR-D^{up}}(t)\}$\\
        FFR Penalty region & $\rho^\text{FFR}(t)$ & (\ref{equ:penalty_regions}) \\
        $\text{FCR-D}^\text{up}$ Penalty region & $\rho^{\text{FCR-D}^\text{up}}(t)$ & (\ref{equ:penalty_region_fcrd}) \\
        Plant transfer function & $G_i(s)$ & $1/\left(\tau_i s + 1\right)$ \\
        DER time constants & $\tau_{\text{wind}}, \tau_{\text{pv}}$ & 2s, 1.5s \\
        DER time constants & $\tau_{\text{bess}}, \tau_c$ & 0.1s, 0.081s \\
        \midrule
        Procurement period & $\mathcal{H}_l\in \mathcal{H}$ & hourly \\
        Time period & $\mathcal{H}$ & 2025-04-06 to 2025-04-12 \\
        \#Scenario forecasts & $K$ & 25 \\
        Market service & $MS, \forall \mathcal{H}_l\in T$ & 50\% FFR + 50\% $\text{FCR-D}^\text{up}$ \\
        Penalty price & $q$ & $3\pi$ \cite{FFR_finland} \\
        Minimum capacity & $b_\text{min}$ & 0.1\,MW \cite{FFR_finland} \\
        \bottomrule
    \end{tabular}
    \vspace{-0.4cm}
\end{table}
Revenues from energy fees do not apply \cite{FCR_finland, FFR_finland} and we ignore other insignificant costs (grid fee, O\&M, etc.).

\subsubsection{Historical Data, Forecasts and DER Models}

Let $\mathcal{H}_l\in \mathcal{H}$ denote the discrete analyzed in the cast study.
\href{https://www.meteoblue.com/en/historyplus}{Meteoblue history+} was employed for the archived historical weather data \cite{meteoblue_history_2024}. It is available for the precisely selected Finnish location and deemed to have high precision. As locational accuracy is a factor of interest, NEMSGLOBAL was chosen as a model. Additional to historical weather data, historical \textit{forecast} data is needed to simulate the DVPP operation. \cite{haessig2015energy} outlines that locational forecast errors can be approximated by an AR(1) model, which we adopt:
\begin{equation}
\begin{gathered}
    e_i(\mathcal{H}_l)=\phi_i \cdot e_i(\mathcal{H}_l-1) + \sigma_i(\mathcal{H}_l), \\
    e_i(\mathcal{H}_l):\text{error at time }\mathcal{H}_l, \\
    \phi_i:\text{autocorrelation coefficient,} \\
    \sigma_i(\mathcal{H}_l): \text{random variance at time }\mathcal{H}_l,
\end{gathered}
\end{equation}
where the parameters $\phi_i$ are fitted for Wind and PV power production, respectively. The Root Mean Square Error is 6.5\% and 4.5\% for Wind and PV, respectively. 

All DERs $i\in\mathcal{N}$ were approximated by a first order transfer function $1/\left(\tau_i s + 1\right)$.

For PV power production, the following model is used:
\[
\overline{p}_\text{pv}(\mathcal{H}_l)=C_{PV}\cdot\frac{\mathrm{GHI}(\mathcal{H}_l)}{1000}\cdot\kappa_{thermal}(\mathcal{H}_l)\cdot \kappa_{snow}(\mathcal{H}_l)
\]
where $C_{PV}=15\text{MW}$, $\mathrm{GHI}(\mathcal{H}_l)$ is solar irradiation, $\kappa_{thermal}(\mathcal{H}_l)$ thermal efficiency, and $\kappa_{snow}(\mathcal{H}_l)$ snow cover. The parameters are calculated from the Meteoblue dataset \cite{meteoblue_history_2024}.

For Wind power production, typical parameters of cut in and cut out speed are used, as shown in Figure~\ref{fig:wind}. The BESS was modeled to have a limited energy capacity of 5 MWh.

\begin{figure}[bt]
    \centering
    \includegraphics[width=.8\linewidth]{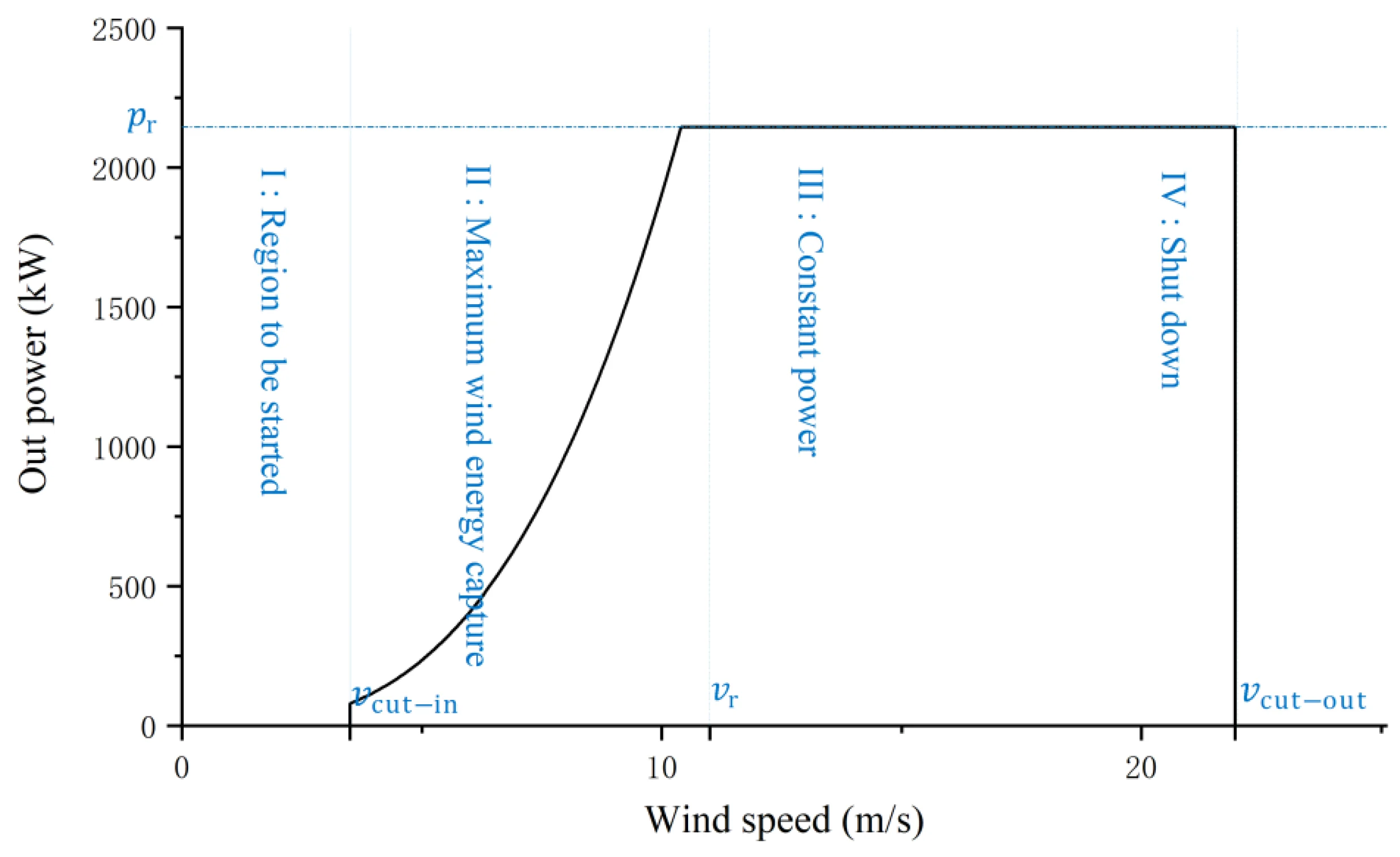}
     \vspace{-.2cm}
    \caption{Wind power production curve used in this work.}
    \label{fig:wind}
        \vspace{-.4cm}
\end{figure}

\bibliographystyle{ieeetr}
\bibliography{references}
\balance

\end{document}

%% file: pics/threeblocks.tex
    \begin{tikzpicture}[
        node distance=0.8cm and 2cm,
        >=Latex,
        font=\footnotesize,
        line/.style={-Latex, thick},
        sum_node/.style={draw, circle, fill=white, inner sep=0pt, minimum size=4mm},
        outerdevice/.style={draw=#1, fill=#1!5, thick, minimum width=6cm, minimum height=2.2cm, rounded corners},
        innerplant/.style={draw=black!70, fill=gray!30, thick, minimum width=2cm, minimum height=0.8cm, align=center, rounded corners=1pt},
        innercontrol/.style={draw=#1!80!black, fill=#1!30, thick, minimum width=2cm, minimum height=0.6cm, rounded corners, align=center},
        dashedbox/.style={draw=gray!60, dashed, rounded corners, thick, inner sep=1.5em},
        branch/.style={fill, shape=circle, minimum size=4pt, inner sep=0pt}
    ]

    \node[input] (inputvec) at (0,0) {};
    
    \coordinate (busStart) at ($(inputvec)+(1.5,0)$);
        
    \draw[thick] (busStart) -- ++(0,-5.07) coordinate (busEnd);

    \node[outerdevice=solarPVColor, right=of busStart, anchor=west] (dev1) {};
    \node[above, inner sep=2pt] at (dev1.north) {$T_1(s)$};

    \node[innerplant] (pre1) at ($(dev1.center)+(-1.3,0.4)$) {$M_\text{PV}(s)\cdot T_\mathrm{des}(s)$};
    \node[innerplant] (plant1) at ($(dev1.center)+(1.3,0.4)$) {PV};
    \node[innercontrol=solarPVColor, below=0.3cm of plant1] (control1) {control};

    \coordinate (splitIn1) at ($(dev1.west) + (0.6, 0)$);
    
    \draw[line] (busStart) |- (pre1.west);
    \draw[line] (pre1.east) |- (plant1.west);

    \coordinate (splitOut1) at ($(plant1.east) + (0.4, 0)$);
    \draw[thick] (plant1.east) -- (splitOut1); 
    \draw[line] (splitOut1) |- (control1.east); 
    \draw[line] (control1.north) -- (plant1.south); 

    \node[sum_node, right=2.5cm of plant1.east] (sum1) {};
    
    \draw[line] (plant1.east) -- node[near end, above] {$\Delta p_1$} (sum1);

    \coordinate (inplabel_start) at ($(plant1.west) + (-6.5, 0)$);
    \coordinate (inplabel_end) at ($(busStart) + (0, .4)$);
    \draw[line] (inplabel_start) -- node[pos=0, above] {$\Delta f$} (inplabel_end);

    \node[outerdevice=WindColor, below=0.5cm of dev1] (dev2) {};
    \node[above, inner sep=2pt] at (dev2.north) {$T_2(s)$};

    \node[innerplant] (pre2) at ($(dev2.center)+(-1.3,0.4)$) {$M_\text{Wind}(s)\cdot T_\mathrm{des}(s)$};
    \node[innerplant] (plant2) at ($(dev2.center)+(1.3,0.4)$) {Wind};
    \node[innercontrol=WindColor, below=0.3cm of plant2] (control2) {control};

    \draw[line] (busStart |- pre2.west) -- (pre2.west);
    \draw[line] (pre2.east) |- (plant2.west);

    \coordinate (splitOut2) at ($(plant2.east) + (0.4, 0)$);
    \draw[line] (splitOut2) |- (control2.east); 
    \draw[line] (control2.north) -- (plant2.south);

    \node[sum_node, right=2.5cm of plant2.east] (sum2) {};

    \draw[line] (plant2.east) -- node[near end, above] {$\Delta p_2$} (sum2);

    \node[outerdevice=BESSColor, below=0.5cm of dev2] (dev3) {};
    \node[above, inner sep=2pt] at (dev3.north) {$T_3(s)$};

    \node[innerplant] (pre3) at ($(dev3.center)+(-1.3,0.4)$) {$M_\text{BESS}(s)\cdot T_\mathrm{des}(s)$};
    \node[innerplant] (plant3) at ($(dev3.center)+(1.3,0.4)$) {BESS};
    \node[innercontrol=BESSColor, below=0.3cm of plant3] (control3) {control};

    \draw[line] (busStart |- pre3.west) -- (pre3.west);
    \draw[line] (pre3.east) |- (plant3.west);

    \coordinate (splitOut3) at ($(plant3.east) + (0.4, 0)$);
    \draw[line] (splitOut3) |- (control3.east);
    \draw[line] (control3.north) -- (plant3.south);

    \node[sum_node, right=2.5cm of plant3.east] (sum3) {};

    \draw[line] (plant3.east) -- node[near end, above] {$\Delta p_3$} (sum3);

    
    \draw[line] (sum3) -- (sum2) node[midway, right] {\tiny +};
    \draw[line] (sum2) -- (sum1) node[midway, right] {\tiny +};
    \node[above=0.03cm of sum1] {+}; 

    \coordinate (aggEnd) at ($(sum1.east)+(1.5,0)$);
    \draw[line] (sum1.east) -- node[midway, above] {$\Delta p_{\text{agg}}$} (aggEnd);

    \begin{scope}[on background layer]
        \node[dashedbox, draw=orgColor!80, fill=orgColor!5, fit=(dev1) (dev3), inner xsep=1em, label={[orgColor!80, font=\small, yshift=1ex]north:set of (controllable) DERs in DVPP}] (setControllable) {};
    \end{scope}

    \node[branch] at (inplabel_end) {};
    
    \node[branch] at (busStart |- pre2.west) {};
    
    \node[branch] at (busStart |- pre3.west) {};

    \end{tikzpicture}